\documentclass[a4paper]{article}

\usepackage[english]{babel}
\usepackage[utf8]{inputenc}
\usepackage[T1]{fontenc}
\usepackage{amsmath}
\usepackage{graphicx}
\usepackage[colorinlistoftodos]{todonotes}
\usepackage[a4paper,left=2.5cm,right=2.5cm,top=2.cm,bottom=2cm]{geometry}
\usepackage{authblk}
\usepackage{dsfont}
\usepackage{txfonts}
\usepackage{bbold} 
\usepackage{amsopn}  
\usepackage{amssymb}
\usepackage{graphicx}
\usepackage{array}
\usepackage{soul}
\usepackage{float}
\usepackage{xparse}
\usepackage{multirow}
\usepackage{pifont}
\newcommand{\xmark}{\ding{55}}%
\newcolumntype{M}[1]{>{\centering\arraybackslash}m{#1}}
\newtheorem{example}{Example}
\newtheorem{proposition}{Proposition}
\newtheorem{definition}{Definition}
\newtheorem{theorem}{Theorem}

\newtheorem{corollary}{Corollary}
\newenvironment{proof}{\noindent\textit{Proof~~}}
{\nolinebreak[4]\hfill$\blacksquare$\\\par}
\title{Nash equilibria in four-strategy quantum game extensions of the Prisoner's Dilemma}

\author{
  Piotr Frąckiewicz\\
  \textit{Institute of Exact and Technical Sciences, Pomeranian University in Słupsk, Poland} \\
  piotr.frackiewicz@upsl.edu.pl
  \and
  Anna Gorczyca-Goraj\\
  \textit{Department of Operations Research,
        University of Economics in Katowice, Poland} \\
        anna.gorczyca-goraj@uekat.pl
    \and
Krzysztof Grzanka\\
  \textit{Department of Operations Research,
        University of Economics in Katowice, Poland} \\
        krzysztof.grzanka@uekat.pl
    \and
Katarzyna Nowakowska\\
  \textit{Institute of Exact and Technical Sciences, Pomeranian University in Słupsk, Poland} \\
  katarzyna.nowakowska@upsl.edu.pl
    \and 
  Marek Szopa\\
  \textit{Department of Operations Research,
        University of Economics in Katowice, Poland} \\
        marek.szopa@uekat.pl
}

\date{\today}

\begin{document}
\maketitle

\begin{abstract}
 This paper investigates Nash equilibria in pure strategies for quantum approach to the Prisoner’s Dilemma. The quantization process involves extending the classical game by introducing two additional unitary strategies. We consider five classes of such quantum games, which remain invariant under isomorphic transformations of the classical game. For each class, we identify and analyze all possible Nash equilibria. Our results reveal the complexity and diversity of strategic behavior in the quantum setting, providing new insights into the dynamics of classical decision-making dilemmas. In the case of the standard Prisoner's Dilemma, the resulting Nash equilibria of quantum extensions are found to be closer to Pareto optimal solutions than those of the classical equilibrium.
 \vspace{0.5cm}

\end{abstract}
\section{Introduction}
A principal objective of quantum game theory is to establish a methodology for transforming classical game theory problems into a quantum mechanical framework \cite{eisert_quantum_1999,flitney_introduction_2002,meyer_quantum_1999,li_continuous-variable_2002}. Subsequently, the characteristics of the resulting game are analysed using techniques from classical game theory  \cite{frackiewicz_quantum_2021,fadaki_quantum_2022,iqbal_evolutionarily_2001,aoki_repeated_2020}, or the quantum game is examined in terms of concepts from quantum computing \cite{chen_quantum_2003,benjamin_thermodynamic_2020,altintas_prisoners_2022,iqbal_equivalence_2018}. Similar to classical game theory, the fundamental problem explored in quantum game theory is the problem of finding rational strategy profiles and answering the question of whether the quantum extension of the game affects the final outcome \cite{flitney_nash_2007,frackiewicz_remarks_2016,szopa_efficiency_2021,landsburg_nash_2011,bolonek-lason_mixed_2017}. In this context, the Nash equilibrium (NE) is a widely used solution concept \cite{nash_equilibrium_1950} that is considered a necessary condition for a given strategy profile to be considered rational. A NE is defined as a strategy profile where no player can improve their payoff by changing their own strategy, assuming that all other players' strategies remain unchanged. This solution concept is applicable in both classical and quantum games due to the way it is formulated.

The Prisoner's Dilemma (PD) is a classical problem in game theory, illustrating the conflict between individual rationality and collective welfare \cite{axelrod_evolution_2006}. Traditionally, the PD game is defined for two players, each having two strategies: cooperate or defect. The standard PD game has a single NE where both players choose to defect, leading to a suboptimal outcome for both. However, the introduction of quantum strategies offers new possibilities for altering this equilibrium structure, potentially allowing for outcomes that are more beneficial to all players involved \cite{eisert_quantum_1999,li_entanglement_2015,szopa_how_2014}.

In our previous works, we explored the quantum extensions of classical games using the Eisert-Wilkens-Lewenstein (EWL) scheme, which introduces additional unitary strategies alongside the classical strategies \cite{frackiewicz_permissible_2024-1,frackiewicz_permissible_2024}. The quantum extensions have been classified into a number of distinct categories, based on the characteristics of the admissible quantum strategies that preserve invariance with respect to isomorphic transformations of the classical game. Our focus was primarily on identifying the conditions under which these quantum games preserve the structural characteristics of the original game while extending the strategic landscape available to the players.

The main goal of the present paper is to examine these extensions in more detail by identifying every NE in the pure strategy profiles of the quantum-enhanced PD. By examining the PD in its most generalized form, considering any admissible set of payoffs, the study seeks to understand under what conditions NE can be achieved for each class of extension. Specifically, the work investigates the constraints that the payoff matrix must satisfy for a given pure strategy profile to be considered a NE across all identified classes of quantum extensions. 

Through this comprehensive analysis, we provide a detailed characterization of strategy profiles and their corresponding NE, thereby extending the understanding of quantum strategies' impact on traditional game-theoretical problems. The paper demonstrates that the NE obtained are more closely aligned with Pareto-optimal solutions than the classical PD. Nevertheless, the class of extensions under consideration does not encompass fully cooperative equilibria. One example of this is the so-called "magic strategy" \cite{eisert_quantum_1999}, which, however, does not satisfy the condition of independence from isomorphic transformations of the classical game \cite{frackiewicz_permissible_2024-1}, which is a necessary condition for the extension to be unambiguous. This contribution not only enhances the theoretical framework of quantum game theory but also has potential applications in fields such as quantum computing and strategic decision-making \cite{fadaki_quantum_2022,huang_fight_2024}, where understanding complex interactive dynamics is crucial.

The work is divided into 5 parts. In the second section we briefly define the key concepts of PD, NE, the EWL quantum game scheme, and prove that positive affine transformations of the payoffs of the classical game do not affect the preference relations of the quantum game. In the third section, we recall five classes of quantum extensions of the classical $2\times 2$ game. In these extensions, quantum players have two additional unitary strategies in addition to their initial classical strategies. The extensions are invariant to isomorphic transformations of the classical game \cite{frackiewicz_permissible_2024}. Furthermore, we demonstrate the symmetry of quantum extensions of the symmetric game. In the fifth section, which is divided into five sub-sections, we analyse the existence of NE of successive classes of extensions. To do so, we examine each of 16 possible profiles of pure strategies. Where equilibria exist, we give the conditions that must be satisfied by the parameters of quantum strategies and PD payoffs. 

\section{Preliminaries}
The games we focus on in our research, including both classical and quantum games, are classified as strategic form games. Formally, a game in strategic (normal) form is defined as follows \cite{maschler_game_2020}:
\begin{definition}
A game in strategic form is a triple $\mathcal{G} = (N, (S_{i})_{i\in N}, (u_{i})_{i\in N})$ in which 
\begin{enumerate}
\item[(i)] $N = \{1,2, \dots, p\}$ is a finite set of players;
\item[(ii)] $S_{i}$ is the set of strategies of player $i$, for each player $i\in N$;
\item[(iii)] $u_{i}\colon S_{1}\times S_{2} \times \cdots \times S_{p} \to \mathds{R}$ is a function that relates each vector of strategies $s = (s_{i})_{i\in N}$ to the payoff $u_{i}(s)$ of the player $i$, for each player $i\in N$.
\end{enumerate}
\end{definition}
A strategic form finite game involving two players can be represented by a bimatrix:
\begin{equation}\label{generalbimatrix}
\Delta = 
\begin{pmatrix}
(\Delta^1_{11}, \Delta^2_{11}) & (\Delta^1_{12}, \Delta^2_{12}) & \cdots &(\Delta^1_{1m}, \Delta^2_{1m}) \\ 
(\Delta^1_{21}, \Delta^2_{21}) & (\Delta^1_{22}, \Delta^2_{22}) & \cdots &(\Delta^1_{2m}, \Delta^2_{2m}) \\ 
\vdots & \vdots & \ddots & \vdots \\ 
(\Delta^1_{n1}, \Delta^2_{n1}) & (\Delta^1_{n2}, \Delta^2_{n2}) & \cdots &(\Delta^1_{nm}, \Delta^2_{nm})
\end{pmatrix} = (\Delta^1, \Delta^2).
\end{equation}
The interpretation of such a notation is that player 1 (the row player) chooses row $i\in S_{1}$ from his set of strategies $S_{1} = \{1, \dots, n\}$, and player 2 (the column player) chooses column $j\in S_{2}$ from her set $S_{2} = \{1, \dots, m\}$. The combination of player 1 using strategy $i$ and player 2 using strategy $j$ will be represented as the ordered pair $(i,j)$ and referred to as a strategy profile. As the result of the game, player 1 receives payoff $u_1(i,j)=\Delta^1_{ij}$ and player 2 receives $u_2(i,j)=\Delta^2_{ij}$. Taking into account the elements that define a game in strategic form, we can identify the payoff function of (\ref{generalbimatrix}) as matrices $\Delta^1=(\Delta^1_{ij})$ and $\Delta^2=(\Delta^2_{ij})$ and denote the game (\ref{generalbimatrix}) as $(\Delta^1,\Delta^2)$. 
Among the games (\ref{generalbimatrix}), we can distinguish those that have certain special characteristics. An example of such games are the symmetric games \cite{van_damme_stability_1987}.
\begin{definition}\label{symdefinition}
Let $\mathcal{G} = (N, (S_{1}, S_{2}), (u_{1}, u_{2}))$ be a two-player finite strategic game. $\mathcal{G}$ is said to be symmetric if $S_{1}=S_{2}$ and $u_{1}(s_{1}, s_{2}) = u_{2}(s_{2}, s_{1})$ for all $s_{1} \in S_{1}, s_{2} \in S_{2}$.
\end{definition}
In matrix notation, the fact that a game $(\Delta^1,\Delta^2)$ is symmetric means that $\Delta^2 = (\Delta^1)^{T}$. 

One of the most well-known symmetric games is the PD. It is a two-player game that can be represented by a $2\times 2$ bimatrix in the form of
\begin{equation}\label{PDgame}
\begin{pmatrix}
         (R,R) & (S,T)\\
         (T,S) & (P,P)
        \end{pmatrix}
 ~\text{where}~ T>R>P>S ~\text{and}~ 2R>T+S.
\end{equation}
In the realm of game theory, the notion of NE plays a key role as a fundamental solution concept \cite{nash_non-cooperative_1951}. This equilibrium represents a strategy profile such that no player can gain a better payoff by deviating from her equilibrium strategy, provided that the other players' strategies remain unchanged. NE provides players with a certain level of stability within a game as in a NE, each player's strategy is a best response to the strategies of the other players. 

In the literature, we find many ways to express NE depending on the type of game being considered \cite{maschler_game_2020}, \cite{osborne_course_1994}. In what follows, we formulate NE in terms of the games we examine in this work, specifically focusing on pure NE in bimatrix games.
\begin{definition}
A strategy profile $(i^*, j^*)$ is a (pure) NE if $\Delta^1_{i^*j^*} \geq \Delta^1_{ij^*}$ for every $i \in S_{1}$ and $\Delta^2_{i^*j^*} \geq \Delta^2_{i^*j}$ for every $j\in S_{2}$.
\end{definition}
For example, it can be easily verified that the unique NE in (\ref{PDgame}) is strategy vector $(2,2)$ that gives the payoff of $P$ for each player. 

Directly from Definition~\ref{symdefinition}, it also follows that there is a symmetry of the set of NE of a two-player symmetric game: if $(s_{1}, s_{2})$ is a NE then $(s_{2}, s_{1})$ is also a NE.
Now, we review the Eisert-Wilkens-Lewenstein scheme for $2\times 2$ bimatrix games. The presented approach is based on \cite{frackiewicz_quantum_2021}.

Let us consider a basis $\{|\psi_{kl}\rangle\colon k,l \in \{1,2\}\}$ of $\mathds{C}^2\otimes \mathds{C}^2$ such that
\begin{equation}
\begin{aligned}
&|\psi_{11}\rangle = \frac{1}{\sqrt{2}}\left(|00\rangle + i|11\rangle\right), \quad |\psi_{12}\rangle = \frac{1}{\sqrt{2}}\left(i|01\rangle -|01\rangle\right),\\ &|\psi_{21}\rangle = -\frac{1}{\sqrt{2}}\left(|01\rangle -i|10\rangle\right), \quad |\psi_{22}\rangle = -\frac{1}{\sqrt{2}}\left(i|00\rangle + |11\rangle\right).
\end{aligned}
\end{equation}
The Eisert-Wilkens-Lewenstein approach to game (\ref{generalbimatrix}) for $S_{1} = S_{2} = \{1,2\}$ is defined by a triple $(\{1,2\}, \{T_{1}, T_{2}\}, \{v_{1}, v_{2}\})$ where
\begin{itemize}
\item $\{1,2\}$ is the set of players,
\item $T_{i}$ is a set of unitary operators from $\mathsf{SU}(2)$ with elements 
\begin{equation}\label{qstrategy}
U_{i}(\theta_{i}, \alpha_{i}, \beta_{i}) = \begin{pmatrix}
e^{i\alpha_{i}}\cos\frac{\theta_{i}}{2} & ie^{i\beta_{i}}\sin\frac{\theta_{i}}{2} \\ 
ie^{-i\beta_{i}}\sin\frac{\theta_{i}}{2} & e^{-i\alpha_{i}}\cos\frac{\theta_{i}}{2}
\end{pmatrix}, \quad \theta_{i} \in [0,\pi] ~\text{and}~ \alpha_{i}, \beta_{i} \in [0,2\pi),
\end{equation}
\item $v_{i}\colon T_{1} \times T_{2} \to \mathds{R}$ is the $i$-th player payoff function given by
\begin{align}
&v_{i}(U_{1}, U_{2}) = \sum_{k,l \in \{1,2\}}\Delta^i_{kl}|\langle \psi_{kl}|U_{1}\otimes U_{2}|\psi_{11}\rangle|^2, \label{v1}
\end{align}
where $\Delta^i_{kl}$ are the payoffs of the classical $2\times 2$ bimatrix game (\ref{generalbimatrix}). 
\end{itemize}
Using formula (\ref{v1}) we can determine the explicit form of the pair of players' payoffs,
\begin{align}\label{generalEWLpayoff}
    &(u_{1}, u_{2})(U_{1}(\theta_{1}, \alpha_{1}, \beta_{1}), U_{2}(\theta_{2}, \alpha_{2}, \beta_{2})) \nonumber\\ &\quad = (\Delta^1_{11}, \Delta^2_{11})\left(\cos(\alpha_{1} + \alpha_{2})\cos\frac{\theta_{1}}{2}\cos\frac{\theta_{2}}{2} + \sin(\beta_{1} + \beta_{2})\sin\frac{\theta_{1}}{2}\sin\frac{\theta_{2}}{2}\right)^2 \nonumber\\
    &\quad + (\Delta^1_{12}, \Delta^2_{12})\left(\cos(\alpha_{1} - \beta_{2})\cos\frac{\theta_{1}}{2}\sin\frac{\theta_{2}}{2} + \sin(\alpha_{2} - \beta_{1})\sin\frac{\theta_{1}}{2}\cos\frac{\theta_{2}}{2}\right)^2 \nonumber \\ 
    &\quad + (\Delta^1_{21}, \Delta^2_{21})\left(\sin(\alpha_{1} - \beta_{2})\cos\frac{\theta_{1}}{2}\sin\frac{\theta_{2}}{2} + \cos(\alpha_{2} - \beta_{1})\sin\frac{\theta_{1}}{2}\cos\frac{\theta_{2}}{2}\right)^2 \nonumber \\ 
    &\quad + (\Delta^1_{22}, \Delta^2_{22})\left(\sin(\alpha_{1} + \alpha_{2})\cos\frac{\theta_{1}}{2}\cos\frac{\theta_{2}}{2} - \cos(\beta_{1} + \beta_{2})\sin\frac{\theta_{1}}{2}\sin\frac{\theta_{2}}{2}\right)^2.
\end{align}
In our work, we will apply certain important and useful aspects of John von Neumann's utility theory. This theory provides us with a tool that allows us to classify games. The payoff functions of all games within a class determine the same preference relations of the players. As a result, for a given strategy of an opponent, the best response of the player is the same in every game within the class, and consequently, the games are equivalent with respect to NE.
\begin{definition}\cite{maschler_game_2020}
    Let $u\colon X \to \mathds{R}$ be a function. A function $v\colon X \to \mathds{R}$ is a positive affine transformation of $u$ if there exists a positive real number $\lambda >0$ and a real number $\mu$ such that 
    \begin{equation}
v(x) = \lambda u(x) + \mu, \quad \forall_{x \in X}.
    \end{equation}
\end{definition}
A special case of von Neumann's linear utility function theorem is 
\begin{theorem}
If $u_{i}$ is a payoff function representing player $i$-th preference relation then every positive affine transformation of $u_{i}$ is a payoff function representing the same preference relation.
\end{theorem}
Let us consider a general PD game given by (\ref{PDgame}). Let us define a positive affine transformation of the form
\begin{equation}
f(x) = \frac{1}{T-S}(x-S).
\end{equation}
This transformation permits the payoffs of the general form of PD (\ref{PDgame}) to be reduced to two parameters, $r$ and $p$, with values in the interval $[0,1]$:
\begin{align}
f(S) &= \frac{1}{T-S}(S-S) = 0, \label{f(s)}\\
f(T) &= \frac{1}{T-S}(T-S) = 1, \\
f(R) &= \frac{1}{T-S}(R-S) = r, \\
f(P) &= \frac{1}{T-S}(P-S) = p, \label{f(p)}
\end{align}
and $0<p<r<1$. As a result, we obtain a game
\begin{equation}\label{prisonerdilemma}
\begin{pmatrix}
         (f(R),f(R)) & (f(S),f(T))\\
         (f(T),f(S)) & (f(P),f(P))
        \end{pmatrix},
\end{equation}
which is equivalent to game (\ref{PDgame}) with respect to preference relations. In other words, they represenent the same problem from a game theory point of view. Taking into account (\ref{f(s)})-(\ref{f(p)}) and (\ref{prisonerdilemma}) one can consider a general PD game as 
\begin{equation}\label{afinematrix}
\Gamma=
\begin{pmatrix}
         (r,r) & (0,1)\\
         (1,0) & (p,p)
\end{pmatrix}, \quad 0<p<r<1 \quad \text{and} \quad r > \frac{1}{2}.
\end{equation}
\begin{example}
A commonly used bimatrix of the PD 
\begin{equation}
\begin{pmatrix}\label{pdgame}
(3,3) & (0,5) \\ 
(5,0) & (1,1)
\end{pmatrix}
\end{equation}
is equivalent to the game (\ref{afinematrix}), where $r = 3/5$ and $p = 1/5$. In the remainder of this paper, we investigate NE by proving theorems about the conditions for their existence for the general form of the PD (\ref{afinematrix}). However, for purposes of clarity, selected examples of equilibria will be presented in the context of its common form (\ref{pdgame}).  
\end{example}
The application of a positive affine transformation in the classical game does not affect also the EWL quantum extension of the game.
\begin{proposition}
The payoffs preference relations of the EWL scheme are invariant with respect to a positive affine
transformation of payoffs in the classical game.
\end{proposition}
\begin{proof}
Let us consider a positive affine
transformation $y = \lambda x + \mu$ and a pair of bimatrix games of the form
\begingroup 
\setlength\arraycolsep{4.5pt}
\begin{equation}
\Theta_1 = \begin{pmatrix}
\Delta_{00} & \Delta_{01} \\
\Delta_{10} & \Delta_{11}
\end{pmatrix},\qquad \Theta_2 =\begin{pmatrix}
\lambda\Delta_{00}+\mu & \lambda\Delta_{01} + \mu\\
\lambda \Delta_{10} + \mu & \lambda \Delta_{11} + \mu
\end{pmatrix}.
\end{equation}
\endgroup
Let $(U_{1}, U_{2})$ be a strategy profile that is more preferred by player $i$ than a profile $(U'_{1}, U'_{2})$. Both strategy profiles determine some probability distributions $(p_{kl})$ and $(p'_{kl})$ defined by payoff function in the EWL scheme (\ref{generalEWLpayoff}) for $\Theta_1$, i.e. over $\{\Delta_{kl}, k,l=1,2\}$, and
\begin{equation}
\sum_{k,l = 1,2}p_{kl}\Delta^i_{kl} \geq \sum_{k,l = 1,2}p'_{kl}\Delta^i_{kl}.
\end{equation}
On the other hand in the EWL scheme for $\Theta_2$ 
\begin{align}
    &\sum_{k,l=1,2}p_{kl}(\lambda \Delta^i_{kl} + \mu) - \sum_{k,l=1,2}p'_{kl}(\lambda \Delta^i_{kl} + \mu) \nonumber\\ 
    &\quad = \sum_{k,l=1,2}p_{kl}\lambda \Delta^i_{kl} +\sum_{k,l=1,2}p_{kl} \mu - \sum_{k,l=1,2}p'_{kl}\lambda \Delta^i_{kl} -\sum_{k,l=1,2}p'_{kl} \mu \nonumber \\ &\quad = \lambda \left( \sum_{k,l=1,2}p_{kl} \Delta^i_{kl} - \sum_{k,l=1,2}p'_{kl} \Delta^i_{kl} \right) \geq 0.
\end{align}
\end{proof}
Therefore the strategy profile $(U_1,U_2)$ is more preferred then a profile $(U'_1,U'_2)$ by player $i$ also in the EWL scheme of the game $\Theta_2$. 
\section{Permissible four-strategy quantum extensions}
In \cite{frackiewicz_permissible_2024} it was shown that there are five classes of extensions of classical $2\times 2$ games which satisfy the invariance criteria with respect to the isomorphic transformation of the game. Each of the classical game extension classes below corresponds to the specific parameters $(\theta_i, \alpha_i, \beta_i)$ of the $i\in \{1,2\}$ players' quantum strategies. Since the focus of this paper is on NE, in the following we will only give selected strategy parameters, e.g. those on which the payoffs of a quantum game depend. Further details regarding the remaining parameters of the strategy can be found in Table 1 of the article \cite{frackiewicz_permissible_2024}. 
The objective of this paper is to examine the NE for the PD. Consequently, all subsequent expressions for the expansion matrices will employ the specific matrix $\Gamma$ (\ref{afinematrix}), which represents the general form of this game.
The first extension class $A$ is defined by matrices
\begingroup 
\setlength\arraycolsep{4.5pt}
\renewcommand\arraystretch{1.5}
\begin{equation} \label{klasaA}
\medmuskip = 0.2mu
A_1 = \begin{pmatrix}
\Gamma & a_1\Gamma+a_1'\Gamma_3 \\ 
a_1\Gamma+a_1'\Gamma_3 & b_1\Gamma+b_1'\Gamma_3
\end{pmatrix}, \quad 
A_2 = \begin{pmatrix}
\Gamma & a_2\Gamma_2+a_2'\Gamma_1 \\ 
a_2\Gamma_1+a_2'\Gamma_2 & b_2\Gamma_3+b_2'\Gamma
\end{pmatrix},
\end{equation}
\endgroup
where 
\begin{equation}\label{Gamma123}
 \Gamma_1=\begin{pmatrix}
         (1,0) & (p,p)\\
         (r,r) & (0,1)
        \end{pmatrix},
 \hspace{3mm}
\Gamma_2=\begin{pmatrix}
         (0,1) & (r,r)\\
         (p,p) & (1,0)
        \end{pmatrix},
\hspace{3mm}
 \Gamma_3
=\begin{pmatrix}
        (p,p) & (1,0)\\
         (0,1) & (r,r)
        \end{pmatrix}
\end{equation}
are derived from the PD matrix (\ref{afinematrix}), where the rows, columns, or both have been swapped, $a_i=\cos^{2}\alpha_i$, $a'_i= 1-a_i =\sin^{2}\alpha_i$ and $b_i=\cos^{2}2\alpha_i$, $b'_i= 1-b_i =\sin^{2}2\alpha_i$. Other parameters of quantum strategies are defined in \cite{frackiewicz_permissible_2024}, in particular $\theta_1 = 0$ and $\theta_2 =\pi$ for $A_1$ and vice versa for $A_2$. The second class of extensions $B$, for which $\theta_1=\theta_2=\frac{\pi}{2}$, is defined by the matrix
\begingroup 
\setlength\arraycolsep{4.5pt}
\renewcommand\arraystretch{1.5}
\begin{equation} \label{klasaB}
B = \begin{pmatrix}
\medmuskip = 0.2mu
\Gamma & \frac{\Gamma+\Gamma_1+\Gamma_2+\Gamma_3}{4} \\ 
\frac{\Gamma+\Gamma_1+\Gamma_2+\Gamma_3}{4} & \frac{\Gamma+\Gamma_1+\Gamma_2+\Gamma_3}{4}
\end{pmatrix}. 
\end{equation}
\endgroup
Extension of the class $C$ is given by the formula
\begingroup 
\setlength\arraycolsep{4.5pt}
\renewcommand\arraystretch{1.5}
\begin{equation}
\medmuskip = 0.2mu
C = \begin{pmatrix} \label{klasaC}
\Gamma & t\frac{\Gamma+\Gamma_3}{2}+t'\frac{\Gamma_1+\Gamma_2}{2} \\ 
t\frac{\Gamma+\Gamma_3}{2}+t'\frac{\Gamma_1+\Gamma_2}{2} & t'^2\Gamma + tt'(\Gamma_1+\Gamma_2) + t^2\Gamma_3
\end{pmatrix}, 
\end{equation}
\endgroup
where $t=\cos^{2}\frac{\theta_1}{2}$, $t'= 1-t =\sin^{2}\frac{\theta_1}{2}$. For class $C$, as well as for the other classes $D$ and $E$, $\theta_2=\pi-\theta_1$. The class $D$ can be determined by the matrices:  
\begingroup 
\setlength\arraycolsep{4.5pt}
\renewcommand\arraystretch{1.5}
\begin{equation}\label{klasaD}
\medmuskip = 0.2mu
D_1 = \begin{pmatrix}
\Gamma & t\Gamma+t'\Gamma_2 \\ 
t\Gamma+t'\Gamma_1 & t^2\Gamma + tt'(\Gamma_1+\Gamma_2) + t'^2\Gamma_3
\end{pmatrix}, \quad 
D_2 = \begin{pmatrix}
\Gamma & t\Gamma_3+t'\Gamma_1 \\ 
t\Gamma_3+t'\Gamma_2 & t^2\Gamma + tt'(\Gamma_1+\Gamma_2) + t'^2\Gamma_3
\end{pmatrix}.
\end{equation}
\endgroup
The last class $E$ is determined by the matrices
\begingroup 
\setlength\arraycolsep{4.5pt}
\renewcommand\arraystretch{1.5}
\begin{equation}
\medmuskip = 0.2mu
E_1 = \begin{pmatrix}\label{klasaE}
\Gamma & t\Gamma+t'\Gamma_1 \\ 
t\Gamma+t'\Gamma_2 & t^2\Gamma + tt'(\Gamma_1+\Gamma_2) + t'^2\Gamma_3
\end{pmatrix}, \quad
E_2 = \begin{pmatrix}
\Gamma & t\Gamma_3+t'\Gamma_2 \\ 
t\Gamma_3+t'\Gamma_1 & t^2\Gamma + tt'(\Gamma_1+\Gamma_2) + t'^2\Gamma_3
\end{pmatrix}.
\end{equation}
\endgroup
The analysis of NE will be simplified by the symmetry of the extension matrix. Consequently, we will prove the following theorem. 
\begin{proposition}\label{propsymewl}
If a two player game $\Gamma$ is symmetric then its quantum EWL extension is also a symmetric game.
\end{proposition}
\begin{proof}
First note that 
\begin{equation}
|\langle \psi_{kl}|U_{2} \otimes U_{1}|\psi_{11}\rangle|^2 = |\langle \psi_{lk}|U_{1} \otimes U_{2}|\psi_{11}\rangle|^2
\end{equation}
in the formula (\ref{v1}) for each pair $(k,l) \in \{1,2\}^2$. Moreover, if a bimatrix game $\Gamma$ is symmetric then $\Delta^2_{ij} = \Delta^1_{ji}$. Then, it follows that
\begin{align}
u_{2}(U_{2}, U_{1}) &= \sum_{k,l \in \{1,2\}}\Delta^2_{kl}|\langle\psi_{kl}|U_{2} \otimes U_{1}|\psi_{11}\rangle|^2 \nonumber \\ 
&= \sum_{k,l \in \{1,2\}}\Delta^1_{lk}|\langle\psi_{lk}|U_{1} \otimes U_{2}|\psi_{11}\rangle|^2 = u_{1}(U_{1}, U_{2}).
\end{align}
\end{proof}
From Proposition~\ref{propsymewl}, we can deduce the following fact:
\begin{corollary}
If a two player game $\Gamma$ is symmetric then all extensions $A_1, \ldots, E_2$ are also symmetric games.
\end{corollary}
\begin{example}
As an example, let us examine the symmetries of the $\Gamma$ game (\ref{afinematrix})
\begin{equation}
\Gamma = \begin{pmatrix}
(r,r) & (0, 1) \\
(1,0) & (p,p)
\end{pmatrix}= ( \Gamma^1, \Gamma^2 ).
\end{equation}
It is symmetric, as the players' payoffs submatrices $\Gamma^i$ obey the relation 
\begin{equation}\label{symm1}
\Gamma^2 =
\begin{pmatrix}
         r & 1\\
         0 & p
        \end{pmatrix} = (\Gamma^1)^T.
\end{equation}
In addition, definition (\ref{Gamma123}) allows us to infer that
\begin{equation}\label{symm2}
\Gamma_2^2 =
\begin{pmatrix}
         1 & r\\
         p & 0
        \end{pmatrix} = (\Gamma_1^1)^T,
\hspace{3mm}
\Gamma_2^1 =
\begin{pmatrix}
         0 & r\\
         p & 1
        \end{pmatrix} = (\Gamma_1^2)^T
\hspace{3mm} \text{and} \hspace{3mm}
\Gamma_3^2 = \begin{pmatrix}
         p & 0\\
         1 & r
        \end{pmatrix} =
(\Gamma_3^1)^T.
\end{equation}
The symmetry of the extension matrices (\ref{klasaA}) and (\ref{klasaB})-(\ref{klasaE}) can be attributed to the relationships (\ref{symm1}) and (\ref{symm2}). To illustrate, consider the extension $A_2$
\begingroup 
\setlength\arraycolsep{4.5pt}
\renewcommand\arraystretch{1.5}
\begin{align} \label{KlA2proof}
\medmuskip = 0.2mu
(A_2^1)^T &= \begin{pmatrix}
\Gamma^1 & a_2\Gamma_2^1+a_2'\Gamma_1^1 \\ 
a_2\Gamma_1^1+a_2'\Gamma_2^1 & b_2\Gamma_3^1+b_2'\Gamma^1
\end{pmatrix}^T
=\begin{pmatrix}
(\Gamma^1)^T & (a_2\Gamma_1^1+a_2'\Gamma_2^1)^T \\ \nonumber
(a_2\Gamma_2^1+a_2'\Gamma_1^1)^T & (b_2\Gamma_3^1+b_2'\Gamma^1)^T
\end{pmatrix} = \\
&=\begin{pmatrix}
\Gamma^2 & a_2\Gamma_2^2+a_2'\Gamma_1^2 \\ 
a_2\Gamma_1^2+a_2'\Gamma_2^2 & b_2\Gamma_3^2+b_2'\Gamma^2
\end{pmatrix} = A_2^2.
\end{align}
\endgroup
\end{example}
In the remainder of this analysis, in order to streamline the formulas, we will focus on the extension matrices of the first player, with the understanding that the matrices of the second player are their transpositions.   
\section{Nash equilibria of the quantum extensions of the Prisoner's Dilemma}
The objective of this section is to undertake a comprehensive examination of all extensions of the PD, with a view to identifying NE in pure strategies. For each equilibrium, we will demonstrate which conditions for its existence must be satisfied by the payoffs $r$ and $p$ of the general PD (\ref{afinematrix}) and the parameters $\theta_i$ or $\alpha_i$ of the quantum strategies (\ref{qstrategy}). The parameters $\beta_i$ of the quantum strategy are, for a given extension, each time defined by the parameters $\alpha_i$ \cite{frackiewicz_permissible_2024}.
\subsection{Extension of the $A$ class}
Let $A_1=\left(A_1^1,\left(A_1^1\right)^T\right)$, where
\begin{equation}
\medmuskip = 0.2mu
A_1^1 = \begin{pmatrix}
r & 0 & a_{1}r + a'_{1}p & a'_{1} \\ 
1 & p & a_{1} & a_{1}p + a'_{1}r \\
a_{1}r+a'_{1}p & a'_{1}& b_{1}r+b'_{1}p & b'_{1} \\
a_{1} & a_{1}p + a'_{1}r & b_{1} & b_{1}p + b'_{1}r
\end{pmatrix}.
\end{equation}

Previously defined parameters $a_i, a_i', b_i, b_i'$, for $i=1,2$ can all be expressed through the single parameter $a$ 
\begin{equation}\label{a_parameters}
    \begin{split}
        &a_i  = \cos^{2}(\alpha_i) = a, \\
        &a_i' = \sin^{2}(\alpha_i) = 1 - a, \\
        &b_i  = \cos^{2}(2\alpha_i) = (1-2a)^{2}, \\
        &b_i' = \sin^{2}(2\alpha_i) = 4a(1-a).       
    \end{split}
\end{equation}
Note that $\alpha_i \in [0,2\pi)$ correspond to $a \in [0,1]$. The use of this abbreviated notation will not cause confusion, provided that we remember that in the extension $A_i$ there is always the parameter $a=a_i$. As a result $A_{1}^1$ matrix takes the following form

\begin{equation}
\medmuskip = 0.2mu
A_1^1 = \begin{pmatrix}
r & 0 & ar-ap+p & 1-a\\ 
1 & p & a & ap-ar+r  \\
ar-ap+p & 1-a & r - 4 (a-1)a(p-r) & -4(a-1)a \\
a & ap-ar+r & (1-2a)^2 &  (1-2a)^{2}p - 4(a-1)ar 
\end{pmatrix}.
\end{equation}

\begin{proposition} \label{A_prop1}
  Neither $(1,j)$ nor $(i,1)$, $i,j=1,\ldots,4$, are Nash equilibria.
\end{proposition}

\begin{proof}
    
Let $0 \leq a \leq 1$, $0 < p < r$,  $0.5 < r < 1$.
Note that \begin{equation}\label{a1_11a} 
              ar-ap+p<1 \text{ for every } a\in[0,1].
             \end{equation}
Moreover, 
\begin{equation}\label{a1_11b}
 1-a<ap-ar+r \text{ for } a=1.
\end{equation}
Equations \eqref{a1_11a} and \eqref{a1_11b} lead to conclusion that none of the strategies in the first row and column can be NE.
\end{proof}
\begin{proposition}
The strategy profile (2,2) is a Nash equilibrium for $ 0 < p < r,~$ $\frac{1}{2} < r < 1$, provided $a = 1$.
\end{proposition}
\begin{proof}
 Consider the following set of inequalities
\begin{align}
    & p \geq 0, \label{A22_1} \\
    & p \geq 1-a, \label{A22_2} \\
    & p \geq ap-ar+r. \label{A22_3} 
\end{align}
\end{proof}
Inequality \eqref{A22_3} can be rearranged to the following form:
\begin{equation}
    (p-r)(1-a) \geq 0,
\end{equation}
where it is evident, that the inequality is fulfilled only for $a=1$. Inequalities \eqref{A22_1} and \eqref{A22_2} are also satisfied by $a=1$, which proves that the strategy profile $(2,2)$ is a NE.

\begin{proposition}
Strategy profiles $(2,3)$, $(3,2)$ are NE provided that one of the following four conditions is satisfied
\begin{equation}
    0 < p \leq \frac{1}{6} \wedge \frac{1}{2} < r \leq 1-3p \wedge \frac{1}{4} \leq a \leq \frac{r-1}{r-1-p}
\end{equation}
or
\begin{equation}
0 < p \leq \frac{1}{6} \wedge  1-3p<r<1-p \wedge \frac{p}{1+p-r} \leq a \leq \frac{r-1}{r-1-p}
\end{equation}
or
\begin{equation}
0 < p \leq \frac{1}{6} \wedge  r=1-p \wedge  a = \frac{r-1}{r-1-p}
\end{equation}
or
\begin{equation}\label{propA1_5}
    \frac{1}{6} < p < \frac{1}{2} \wedge \frac{1}{2} < r \leq 1-p \wedge \frac{p}{1+p-r} \leq a \leq \frac{r-1}{r-1-p}.
\end{equation}
Note that if $r=1-p$ in equation (\ref{propA1_5}), then $a = \frac{r-1}{r-1-p}$.
\end{proposition}

\begin{proof}
 Consider the following set of inequalities
 \begin{align}
    & a \geq (1-a)p + ar, \\
    & a \geq 4(a-a^2)p+(1-2a)^2 r, \label{A23_2}\\
    & a \geq (1-2a)^2, \label{A23_3}\\
    & 1-a \geq (1-a)r + ap.
\end{align}
\end{proof}
Inequality \eqref{A23_3} is fulfilled by $a\in \left[\tfrac{1}{4},1\right]$. It is easy to show, that $(1-2a)^2\geq4(a-a^2)p+(1-2a)^2r$ if $a\in[0,1]$ and $(1-a)p+ar\geq 4(a-a^2)p+(1-2a)^2r$ if $a\in\left[\tfrac{1}{4},1\right]$. Thus the above set of inequalities comes down to
\begin{align}
    & a \geq (1-a)p + ar, \label{A23_1}\\
    & \frac{1}{4}\leq a \leq 1,\\
    & 1-a \geq (1-a)r + ap. \label{A23_5}
\end{align}
Inequalities  \eqref{A23_1} and \eqref{A23_5} are satisfied for  $\tfrac{p}{1+p-r}\leq a\leq 1 $ and $0\leq a \leq\tfrac{r-1}{r-1-p}$, respectively. Note that $\frac{p}{1+p-r} \leq \frac{r-1}{r-1-p}$, if $r\leq 1-p$. Moreover, $\frac{p}{1+p-r} > \frac{1}{4}$, if $r>1-3p$ and $0<p\leq \frac{1}{6}$. Thus, if $0 < p \leq \frac{1}{6}$  and $1-3p<r\leq 1-p$, then $\frac{p}{1+p-r} \leq \frac{r-1}{r-1-p}$. If   $0 < p \leq \frac{1}{6}$  and $\frac{1}{2}<r\leq1-3p$, then $\tfrac{1}{4}\leq a\leq\tfrac{r-1}{r-1-p}$. If $\tfrac{1}{6}<p<\tfrac{1}{2}$ and $\tfrac{1}{2}<r<1-p$, then $\frac{p}{1+p-r} \leq \frac{r-1}{r-1-p}$.

\begin{proposition}
Strategy profiles $(2,4)$, $(4,2)$ are NE provided that
\begin{equation}\left( \frac{1}{2}<r<\frac{3-p}{3} \wedge a=1\right) \vee \left(r=\frac{3-p}{3} \wedge a\in\left\{\frac{1}{4}, 1\right\} \right) \vee \left( \frac{3-p}{3}<r<1 \wedge a\in \left[\frac{1-r}{1+p-r},1\right] \right).\end{equation}
\end{proposition}
\begin{proof}
 Consider the following set of inequalities:
 \begin{align}
    & ap-ar+r \geq (1-a), \label{A24_1a}\\
    & ap-ar+r \geq 4(1-a)a, \label{A24_2a}\\
    & ap-ar+r \geq (1-2a)^{2} p - 4(a-1)ar, \label{A24_3a}\\
    & ap-ar+r \geq p. \label{A24_5a}
\end{align}
\end{proof}
Inequality \eqref{A24_5a} is satisfied by $a=1$. The inequality  \eqref{A24_3a} or equivalently  $(r-p)(a-1)(4a-1) \geq 0$ is satisfied by $a \in [0,\frac{1}{4}] \cup \{1\}$. Note that if $a\leq \frac{1}{4}$, inequality \eqref{A24_2a} is also fulfilled.
The solution of the inequality \eqref{A24_1a} is $a \in\left[ \frac{1-r}{1+p-r},1\right]$.
Moreover, notice that $\frac{1-r}{1+p-r}<\frac{1}{4}$, if $r>\frac{3-p}{3}$. Finally, pairs of strategies $(2,4)$ and $(4,2)$ are NE if
\begin{equation}
 \begin{cases}
 a=1 & \text{ for } 0<r<\frac{3-p}{3},\\
 a=1 \vee a=\frac{1}{4} & \text{ for } r=\frac{3-p}{3},\\
 a\in \left[\frac{1-r}{1+p-r},1\right]  & \text{ for } \frac{3-p}{3}<r<1.
\end{cases}
\end{equation}

\begin{proposition}
The strategy profile (3,3) is a Nash equilibrium provided that
\begin{equation}
0<p<\frac{1}{6} \wedge \left( \left( r=1-3p \wedge a=\frac{1}{4}\right) \vee \left(1-3p<r<1 \wedge  \frac{1}{2} - \frac{1}{2}\sqrt{\frac{p}{1+p-r}}\leq a \leq \frac{1}{4}\right)\right)
\end{equation}
or
\begin{equation}
\frac{1}{6}\leq p\leq\frac{1}{2} \wedge \frac{1}{2}<r<1 \wedge  \frac{1}{2} - \frac{1}{2}\sqrt{\frac{p}{1+p-r}}\leq a \leq \frac{1}{4}
\end{equation}
or
\begin{equation}\frac{1}{2}<p < r \wedge p<r<1 \wedge  \frac{1}{2} - \frac{1}{2}\sqrt{\frac{p}{1+p-r}}\leq a \leq \frac{1}{4}.\end{equation}

\end{proposition}
\begin{proof}
Consider the following set of inequalities
 \begin{align}
    & r-4a(a-1)(p-r) \geq ar-ap+p,    \label{A33_1a}\\
    & r-4a(a-1)(p-r) \geq a,          \label{A33_2a}\\
    & r-4a(a-1)(p-r) \geq (1-2a)^{2}.\label{A33_3a}
\end{align}
Inequality \eqref{A33_1a} can be transformed into $(r-p)(a-1)(4a-1) \geq 0$, which is satisfied by $a \in [0,\frac{1}{4}] \cup \{1\}$. Notice that $(1-2a)^2 \geq a$ if $a\leq \frac{1}{4}$. Thus, pair of strategies (3,3) is NE if inequalities \eqref{A33_1a} and \eqref{A33_3a} are satisfied.

Inequality \eqref{A33_3a} is equivalent to $ 4a^2(r-p-1) -4a(r-p-1) + r-1 \geq 0$ and its solution is the interval $\left[\frac{1}{2} -\frac{1}{2}\sqrt{\frac{p}{1+p-r}}, \frac{1}{2}+\frac{1}{2}\sqrt{\frac{p}{1+p-r}} \right]$.

Note that $\frac{1}{2} -\frac{1}{2}\sqrt{\frac{p}{1+p-r}}<\frac{1}{4}$, if $r>1-3p$.
Finally, a pair of strategies $(3,3)$ is NE if the following conditions are satisfied

\begin{equation}
 \begin{cases}
 a=\frac{1}{4} & \text{ for } 0<p<\frac{1}{6} \wedge r=1-3p,\\
 \frac{1}{2} - \frac{1}{2}\sqrt{\frac{p}{1+p-r}}\leq a \leq \frac{1}{4} & \text{ for }  0<p<\frac{1}{6} \wedge 1-3p<r<1,\\
 \frac{1}{2} - \frac{1}{2}\sqrt{\frac{p}{1+p-r}}\leq a \leq \frac{1}{4} & \text{ for } \frac{1}{6}\leq p\leq\frac{1}{2} \wedge \frac{1}{2}<r<1, \\
 \frac{1}{2} - \frac{1}{2}\sqrt{\frac{p}{1+p-r}}\leq a \leq \frac{1}{4} & \text{ for } \frac{1}{2}<p < r \wedge p<r<1.\\
\end{cases}
\end{equation}
\end{proof}

\begin{proposition} \label{prop_A_3_4}
Strategy profiles $(3,4)$ and $(4,3)$ are NE provided that  $a=\frac{1}{4}$, $\frac{1}{2} < r \leq 1-3p$,  $0 < p < \frac{1}{6}.$
\end{proposition}
\begin{proof}
Consider the following set of inequalities
\begin{align}
    &4(1-a)a \geq 1-a, \label{A34_1a} \\
    &4(1-a)a \geq ap-ar+r, \label{A34_2a} \\
    &4(1-a)a \geq (1-2a)^{2}p -4(a-1)ar, \label{A34_3a} \\
    &(1-2a)^2 \geq ar-ap+p, \label{A34_4a} \\
    &(1-2a)^2 \geq a, \label{A34_5a} \\
    &(1-2a)^2 \geq r-4(a-1)a(p-r). \label{A34_6a}
\end{align}
Inequality \eqref{A34_1a} can be written as $ -4a^2 +5a -1 \geq 0$, and so, $a \in \left[ \frac{1}{4},1 \right]$. Inequality \eqref{A34_5a} is equivalent to $ 4a^2 -5a +1 \geq 0$, which is satisfied by $a \in \left[0, \frac{1}{4} \right]$. Moreover, \eqref{A34_3a} is not satisfied if $a=1$. Thus, $a=\frac{1}{4}$.
 Finally, setting $a=\frac{1}{4}$ for the remaining inequalities leads to the conclusion that a pair of strategies $(3,3)$ is NE, if $a=\frac{1}{4}$, $\frac{1}{2}< r\leq1-3p$, $0<p<\frac{1}{6}$.
\end{proof}

\begin{proposition} \label{A_prop6}
The strategy profile $(4,4)$ is a Nash equilibrium provided that one of the following conditions is satisfied
\begin{equation}
    \frac{1}{2}<r\leq\frac{3}{4} \wedge 0<p<r \wedge \frac{1}{2} + \frac{1}{2}\sqrt{\frac{1-r}{p-r+1}} \leq a \leq 1
\end{equation}    
or
\begin{equation}
\frac{3}{4}<r<1\wedge 0<p<3-3r \wedge  \frac{1}{2} + \frac{1}{2}\sqrt{\frac{1-r}{p-r+1}} \leq a \leq 1
\end{equation}
or
\begin{equation}
\frac{3}{4}<r<1\wedge p=3-3r \wedge \left( a=\frac{1}{4} \vee \frac{1}{2} + \frac{1}{2}\sqrt{\frac{1-r}{p-r+1}} \leq a \leq 1 \right)
\end{equation}
or
\begin{equation}
\frac{3}{4}<r<1\wedge 3-3r<p<r \wedge \left( \frac{1}{4} \leq a \leq \frac{1}{2} - \frac{1}{2}\sqrt{\frac{1-r}{p-r+1}} \vee \frac{1}{2} + \frac{1}{2}\sqrt{\frac{1-r}{p-r+1}} \leq a \leq 1 \right).
\end{equation}
\end{proposition}
\begin{proof}
Consider the following set of inequalities
\begin{align}
& (1-2a)^{2}p -4(a-1)ar \geq ap - ar + r, \label{A44_1a} \\
& (1-2a)^{2}p -4(a-1)ar \geq 1-a, \label{A44_2a} \\
& (1-2a)^{2}p -4(a-1)ar \geq 4(1-a)a. \label{A44_3a}
\end{align}
Inequality \eqref{A44_1a} can be transformed to
\begin{equation}
    (p-r)(4a^2 -5a +1) \geq 0
\end{equation}
and its solution is $ \left[ \frac{1}{4}, 1 \right]$. Notice, that if $a\geq\frac{1}{4}$, then $4a(1-a)\geq1-a$. Thus, a pair of strategies $(4,4)$ is NE if inequalities \eqref{A44_1a} and \eqref{A44_3a} are satisfied. Inequality \eqref{A44_3a} is equivalent to

\begin{equation}
    4a^{2}(p-r+1) -4a(p-r+1)+p \geq 0,
\end{equation}
of which the solution is  $\left[\frac{1}{2} - \frac{1}{2} \sqrt{\frac{1-r}{p-r+1}}, \frac{1}{2} + \frac{1}{2} \sqrt{\frac{1-r}{p-r+1}} \right]$.
Note that
\begin{equation}
 \frac{1}{2} + \frac{1}{2} \sqrt{\frac{1-r}{p-r+1}} \leq 1 \text{ for } 0<p<r,
\end{equation}
\begin{equation}
 \frac{1}{2} - \frac{1}{2} \sqrt{\frac{1-r}{p-r+1}} > \frac{1}{4} \text{ for } p>3-3r,
\end{equation}
\begin{equation}
 \frac{1}{2} - \frac{1}{2} \sqrt{\frac{1-r}{p-r+1}} < \frac{1}{4} \text{ for } p<3-3r,
\end{equation}
\begin{equation}
 \frac{1}{2} - \frac{1}{2} \sqrt{\frac{1-r}{p-r+1}} = \frac{1}{4} \text{ for } p=3-3r.
\end{equation}
Finally, a pair of strategies $(4,4)$ is a NE if
\begin{equation}
 \begin{cases}
 \frac{1}{2} + \frac{1}{2}\sqrt{\frac{1-r}{p-r+1}} \leq a \leq 1  & \text{ for } \frac{1}{2}<r\leq\frac{3}{4} \wedge 0<p<r, \\
  \frac{1}{2} + \frac{1}{2}\sqrt{\frac{1-r}{p-r+1}} \leq a \leq 1  & \text{ for }  \frac{3}{4}<r<1\wedge 0<p<3-3r,\\
  a=\frac{1}{4} \vee \frac{1}{2} + \frac{1}{2}\sqrt{\frac{1-r}{p-r+1}} \leq a \leq 1 & \text{ for } \frac{3}{4}<r<1\wedge p=3-3r, \\
 \frac{1}{4} \leq a \leq \frac{1}{2} - \frac{1}{2}\sqrt{\frac{1-r}{p-r+1}} \vee \frac{1}{2} + \frac{1}{2}\sqrt{\frac{1-r}{p-r+1}} \leq a \leq 1 & \text{ for } \frac{3}{4}<r<1\wedge 3-3r<p<r.\\
\end{cases}
\end{equation}
\end{proof}
Let us note that $A_2$ matrix can be obtained from $A_1$ matrix by swapping  3rd and 4th rows and columns. Therefore, by analogy, a similar set of Propositions \ref{A_prop1}-\ref{A_prop6} describing NE can be shown for $A_2$ extensions. A summary of all strategy profiles in the extensions of $A_1$ and $A_2$ for which NE are possible, along with the requirements for payoffs $p$ and $r$, and the parameter $a$, is shown in Table \ref{table_A1}.
\begin{table}
\centering
\renewcommand*{\arraystretch}{2}%
\begin{tabular}{| >{\centering}p{2.5cm} | >{\centering}p{2.5cm} | >{\centering}p{2.5cm} | >{\centering\arraybackslash}  p{5.5cm} | }
\hline
Strategy profile & $ p $ & $ r $ & $ a $ \\
\hline
\hline
$A_{1}(2,2)$ \qquad $A_{2}(2,2)$ & $(0,r)$ & $\left(\frac{1}{2},1\right)$ & \{1\} \\
\hline
\hline
\multirow{3}{2.5cm}{\centering $A_{1}(2,3)$ $A_{1}(3,2)$ $A_{2}(2,4)$ $A_{2}(4,2)$} & \multirow{2}{*}{$ \left( 0, \frac{1}{6} \right] $} & $\left(\frac{1}{2}, 1-3p \right]$ & $\left[\frac{1}{4}, \frac{r-1}{r-1-p} \right]$ \\
\cline{3-4}
& & $(1-3p,1-p]$ & $\left[\frac{p}{1+p-r}, \frac{r-1}{r-1-p}  \right]$\\
\cline{2-4}
& $\left( \frac{1}{6}, \frac{1}{2}\right)$ & $\left( \frac{1}{2}, 1-p \right]$ & $\left[ \frac{p}{1+p-r}, \frac{r-1}{r-1-p} \right]$ \\
\hline
\hline
\multirow{3}{2.5cm}{\centering $A_{1}(2,4)$ $A_{1}(4,2)$ $A_{2}(2,3)$ $A_{2}(3,2)$}& \multirow{3}{*}{$(0,r)$} & $\left(\frac{1}{2}, \frac{3-p}{3} \right)$ & $\left\{1 \right\}$\\
\cline{3-4}
& & $\left\{\frac{3-p}{3} \right\}$ & $\left\{ \frac{1}{4}, 1\right\}$ \\
\cline{3-4}
& & $\left( \frac{3-p}{3}, 1 \right)$  & $\left[ \frac{1-r}{1+p-r},1 \right] $\\
\hline
\hline
\multirow{3}{2.5cm}{\centering $A_{1}(3,3)$ \qquad $A_{2}(4,4)$} & $\left[0,\frac{1}{6} \right]$ & $[1-3p,1)$ & \multirow{3}{*}{$\left[\frac{1}{2} - \frac{1}{2}\sqrt{\frac{1-r}{p-r+1}}, \frac{1}{4} \right]$} \\
\cline{2-3}
& $\left[\frac{1}{6}, \frac{1}{2} \right]$ & $\left( \frac{1}{2}, 1 \right)$ &  \\
\cline{2-3}
& $\left(\frac{1}{2}, r \right)$ & $\left(p,1\right)$ &  \\
\hline
\hline
$A_{1}(3,4)$ $A_{1}(4,3)$ $A_{2}(4,3)$ $A_{2}(3,4)$ & $\left(0, \frac{1}{6} \right)$ & $\left(\frac{1}{2}, 1-3p \right]$ & $\left\{ \frac{1}{4} \right\}$ \\
\hline
\hline
\multirow{4}{2.5cm}{\centering $A_{1}(4,4)$ \qquad $A_{2}(3,3)$}& $(0,r)$ & $\left( \frac{1}{2}, \frac{3}{4} \right]$ & \multirow{2}{*}{$\left[ \frac{1}{2} + \sqrt{\frac{1-r}{p-r+1}} ,1 \right]$}\\
\cline{2-3}
& $(0,3-3r)$ & $\left(\frac{3}{4},1 \right) $ & \\
\cline{2-3}\cline{4-4}
& $\{3-3p \}$ & \multirow{2}{*}{$\left(\frac{3}{4},1 \right) $} & $\left[ \frac{1}{2} + \sqrt{\frac{1-r}{p-r+1}} ,1 \right] \cup \left\{ \frac{1}{4} \right\} $\\
\cline{2-2}\cline{4-4}
& $(3-3r,r)$ &  & $\left[ \frac{1}{4}, \frac{1}{2} - \sqrt{\frac{1-r}{p-r+1}}  \right] \cup \left[ \frac{1}{2} + \sqrt{\frac{1-r}{p-r+1}} ,1 \right]$\\
\hline
\end{tabular}
\caption{Summary of the conditions, for which the given strategy profiles in $A_{1}$ and $A_{2}$ class extensions are NE. For the existence of equilibria, the conjunction of the conditions given in columns p and r (PD payoffs (\ref{afinematrix})) and $a$ (parameter (\ref{a_parameters}) defining unitary strategies (\ref{qstrategy})) must be satisfied.} \label{table_A1}
\end{table}
\begin{example}
Consider the PD given by Eq. \eqref{pdgame}. Below the resulting matrix for the $A_1$ class extension.
\begingroup 
\setlength\arraycolsep{1.7pt}
\renewcommand\arraystretch{1.5}
\begin{align}
&A_{1} =\begin{pmatrix} 
(3,3) & (0,5) & (2a+1,2a+1) & (5-5a,5a) \\
(5,0) & (1,1) & (5a,5-5a) & (3-2a,3-2a) \\
(2a+1,2a+1) & (5-5a,5a)  & (8(a-1)a+3,8(a-1)a+3) & (-20(a-1)a,5(1-2a)^{2}) \\
(5a,5-5a) & (3-2a,3-2a)  & (5(1-2a)^{2},-20(a-1)a) & (1-8(a-1)a,1-8(a-1)a)
\end{pmatrix}.
\label{A1_PD_example}
\end{align}
\endgroup
Note that for pairs of strategies (3,4) and (4,3) the necessary condition $0<p<\frac{1}{6}$ from Proposition \ref{prop_A_3_4} for the existence of a NE is not satisfied since $p=\frac{1}{5}$. 
Table \ref{A1_PD_NE_payoffs} illustrates seven strategy profiles for which NE are feasible, along with the corresponding values of the parameter $a$ that result in maximum equal payoffs for the players.
\begin{table}[ht]
\renewcommand{\arraystretch}{1.7}
\begin{center}
\begin{tabular}{ | >{\centering}p{1.5cm} | >{\centering}p{3.5cm} | >{\centering}p{3.5cm} | >{\centering\arraybackslash}  p{3.5cm} | }
\hline
\xmark & \xmark & \xmark & \xmark\\
\hline
\xmark & (1,1) for $a = 1$ & {$\left(\frac{5}{2}, \frac{5}{2}\right)$ for $a = \frac{1}{2}$}  & (1,1) for $a = 1$  \\
\hline
\xmark & {$\left(\frac{5}{2}, \frac{5}{2}\right)$ for $a = \frac{1}{2}$ } & $\left(\frac{5}{3}, \frac{5}{3}\right)$ for $a = \frac{3-\sqrt{3}}{6}$ & \xmark \\
\hline
\xmark & (1,1)  for $a = 1 $ & \xmark & {$\left(\frac{5}{3},\frac{5}{3}\right)$ for  $a = \frac{3+\sqrt{6}}{6}$}\\
\hline
\end{tabular}
\caption{NE with maximal and equal payoffs and the corresponding $a$ parameters for the $A_1$ class extension of the PD (\ref{pdgame}), the symbol \xmark \, denotes lack of a NE for the corresponding pair of strategies.}\label{A1_PD_NE_payoffs}
\end{center}
\end{table}
\end{example}
Figure \ref{fig1} shows the first player's payoffs for all NE (not necessarily with equal payoffs) of the $A_1$ extension of the PD (\ref{pdgame}) as a function of the parameter $a$.
\begin{figure}
    \centering
    \includegraphics[width=0.9\linewidth]{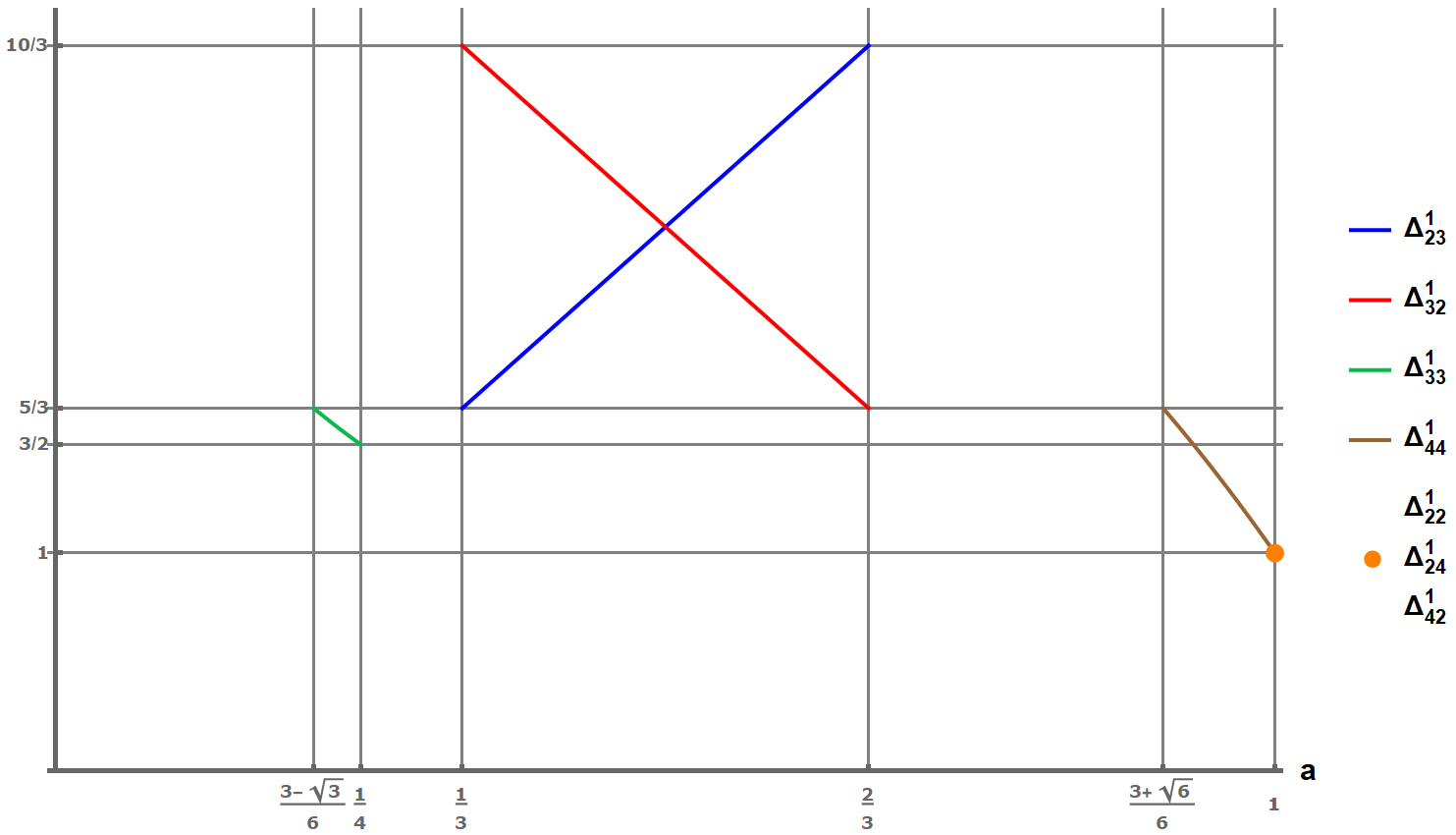}
    \caption{The dependence of NE first player payoffs on the value of the parameter a (in the permissible range) for different strategy profiles of the exemplary PD (\ref{pdgame}) in the extension $A_1$ given by the matrix \eqref{A1_PD_example}. Payoffs $\Delta^1_{22}= \Delta^1_{24}=\Delta^1_{42}=1$, which correspond to NE for $a=1$ are identical and depicted by a single dot.}\label{fig1}
\end{figure}
Figure \ref{fig:A1_PR_payoff_plots} shows the payoffs $\Delta^{i}_{jk}$ for profiles $j \leq k$ of NE in the $A_1$ extension of the PD (\ref{PDgame}) as a function of the payoffs $P$ and $R$ for $S=0$ and $T=5$ and the value of $a$ corresponding to the maximum and equal NE according to Table \ref{A1_PD_NE_payoffs}, where $P=1$ and $R=3$.
\begin{figure}
    \centering
    \includegraphics[width=0.45\linewidth]{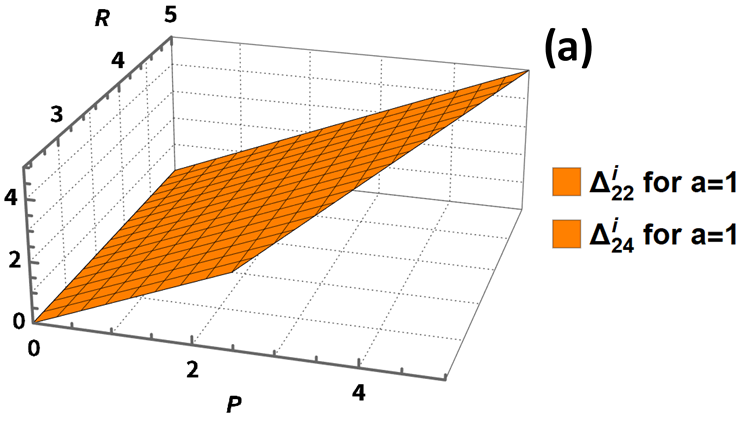} \hspace{0.3cm}
    \includegraphics[width=0.45\linewidth]{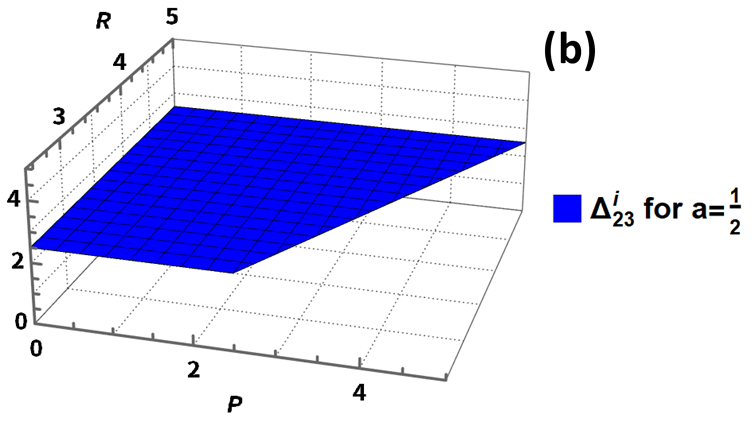}
    \includegraphics[width=0.47\linewidth]{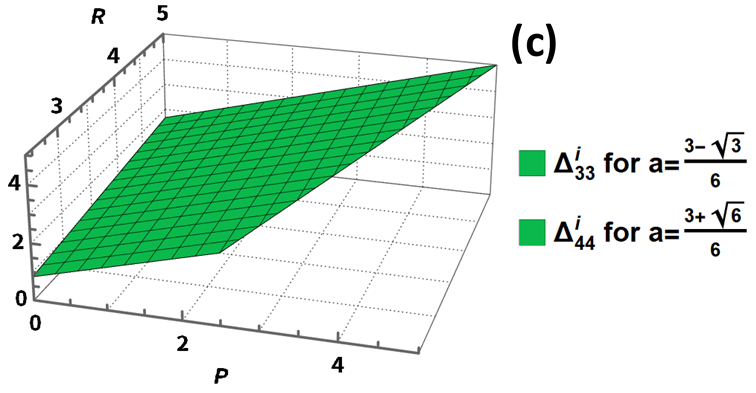}

    \includegraphics[width=0.45\linewidth]{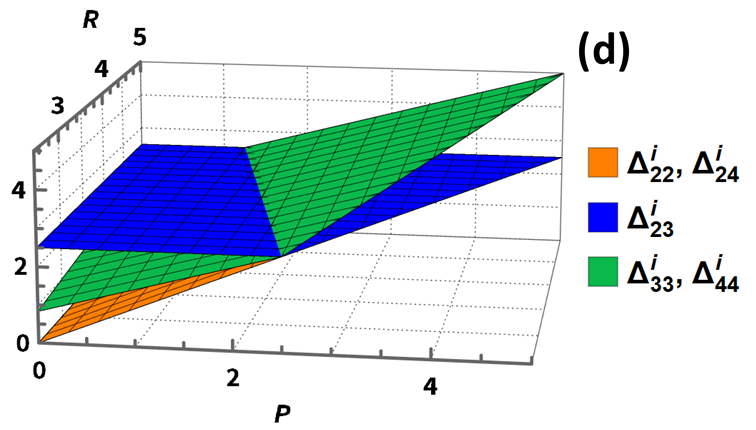} \hspace{0.3cm}
    \includegraphics[width=0.45\linewidth]{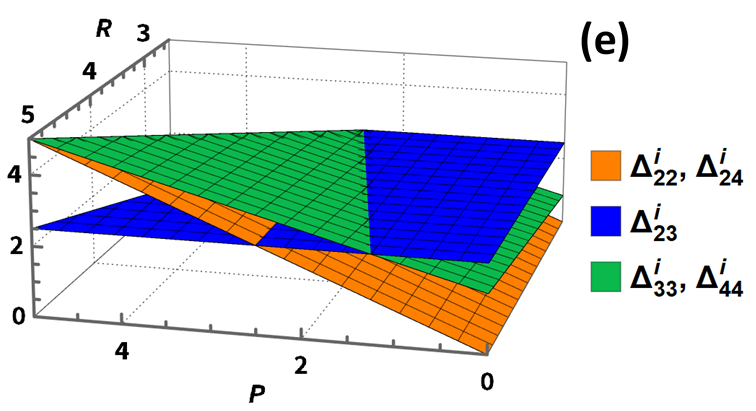}
    \caption{Dependence of the payoffs of the $A_1$ extension of the PD (\ref{PDgame}) on the payoffs $P$ and $R$ for $S=0$ and $T=5$ and the value of $a$ corresponding to the maximum and equal NE according to Table \ref{A1_PD_NE_payoffs}, where $P=1$ and $R=3$. For a better comparison, figures (d) and (e) show the relationships shown in (a), (b) and (c) from two different points of view. In all presented cases the payoffs $\Delta^{i}_{jk}$ are the same for both players $i\in\{1,2\}$}
    \label{fig:A1_PR_payoff_plots}
\end{figure}
\subsection{Extension of the $B$ class}
The $B$ class extension of PD (\ref{afinematrix}) is defined by the first player's payoff matrix:
\begingroup 
\setlength\arraycolsep{1.7pt}
\renewcommand\arraystretch{1.5}
\begin{equation}\label{B_ext}
\medmuskip = 0.2mu
B^1 = \begin{pmatrix}
r & 0 & \frac{1}{4}(1+r+p) & \frac{1}{4}(1+r+p) \\ 
1 & p & \frac{1}{4}(1+r+p) & \frac{1}{4}(1+r+p) \\
\frac{1}{4}(1+r+p) & \frac{1}{4}(1+r+p) & \frac{1}{4}(1+r+p) & \frac{1}{4}(1+r+p) \\
\frac{1}{4}(1+r+p) & \frac{1}{4}(1+r+p) & \frac{1}{4}(1+r+p) & \frac{1}{4}(1+r+p)
\end{pmatrix},
\end{equation}
\endgroup
where  $0<p<r<1$ and $2r > 1$.
\begin{proposition} \label{B_prop}
Depending on the parameters $p$ and $r$, the game defined by the matrix (\ref{B_ext}) exhibits the following Nash equilibria in pure strategies: \begin{itemize}
    \item[i.] The strategy profiles $(1,j)$ and $(i,1)$, for $i,j=1,\ldots,4$ are not NE for any values of $p$ and $r$, 
    \item[ii.] The strategy profile $(2,2)$ is a NE provided that $p \geq \frac{1+r}{3}$,
    \item[iii.] The strategy profile $(2,j)$ and $(i,2)$ for $i,j=3,4$ are NE provided that $p \leq \frac{1+r}{3}$,
    \item[iv.] The strategy profiles $(i,j)$, for $i,j=3,4$ are NE for arbitrary values of $p$ and $r$.
\end{itemize}
\end{proposition}
\begin{proof} The proof is based on the definition of NE and the inequalities that occur for PD defining parameters $r$ and $p$.   
  \begin{itemize} 
      \item[i.] The $(2,1)$ profile is the unique NE candidate in the first column, as player 1 has the highest payoff of $1>r$ and $1>\frac{1+r+p}{4}$, in that column. Nevertheless, player 2's payoff for this profile is $0$, i.e. it is smaller than the payoffs of the other players in this row. Consequently, neither the $(2,1)$ profile nor any other profile in the first column can be classified as NE. The symmetry of the game implies that no profile in the first row can also be classified as NE.
      \item[ii.] The $(2,2)$ is a NE if $p \geq \frac{1+r+p}{4}$ which leads directly to $p \geq \frac{1+r}{3}$
      \item[iii.] The $(3,2)$ is a NE if $\frac{1+r+p}{4} \geq p$ which leads directly to $p \leq \frac{1+r}{3}$. The remaining profiles are also NE, based on the same inequality.
      \item[iv.] For these profiles the payoffs of both players in each row and column are identical, thereby demonstrating that these are NE.
  \end{itemize}  
This completes the proof.
  \end{proof}
  This above proposition is summarised in Table \ref{table_B1}, which shows the conditions that must be met for NE to exist in the respective pure strategy profiles. The payoff values for these equilibria are the same for both players and equal to the corresponding positions of matrix (\ref{B_ext}).
\begin{table}[h] 
\renewcommand{\arraystretch}{1.5}
    \centering
\begin{tabular}{ | >{\centering}p{1.5cm} | >{\centering}p{1.5cm} | >{\centering}p{1.5cm} | >{\centering\arraybackslash}  p{1.5cm} | }
    \hline
 \xmark & \xmark & \xmark & \xmark\\  
\hline
\xmark  & {$p \geq \frac{1+r}{3}$}  & {$p \leq \frac{1+r}{3}$ } & {$p \leq \frac{1+r}{3}$ }   \\
 \hline
\xmark & {$p \leq \frac{1+r}{3}$ } & { $\checkmark$  }  & {$\checkmark$} \\
  \hline
 \xmark & {$p \leq \frac{1+r}{3}$ } & { $\checkmark$  }  &  {$\checkmark$} \\
\hline
\end{tabular}
\caption{The $B$ class strategies parameters resulting in NE. Parameters $p,\,r$ are PD  payoffs, the mark \xmark \, denotes lack of a NE for the corresponding pair of strategies and the mark \checkmark \, denotes that NE exists for all parameter values.}\label{table_B1}
\end{table}
\begin{example}
For the standard PD payoff matrix (\ref{pdgame}), the equivalent game (\ref{afinematrix}) parameters are $r = 3/5$ and $p = 1/5$, therefore $p < \frac{1+r}{3}$. It leads to set of NE strategy profiles $\{(i,j):i\geq 3 \lor j\geq 3\}$ with payoffs all equal to $2\frac{1}{4}$, see Table \ref{Table_B2}.
\begin{table}[H]
\renewcommand{\arraystretch}{1.7}
\begin{center}
\begin{tabular}{ | >{\centering}p{1.5cm} | >{\centering}p{1.5cm} | >{\centering}p{1.5cm} | >{\centering\arraybackslash}  p{1.5cm} | }
\hline
\xmark & \xmark & \xmark & \xmark\\
\hline
\xmark & \xmark & $\left(2\frac{1}{4}, 2\frac{1}{4}\right)$  & $\left(2\frac{1}{4}, 2\frac{1}{4}\right)$ \\
\hline
\xmark & $\left(2\frac{1}{4}, 2\frac{1}{4}\right)$  & $\left(2\frac{1}{4}, 2\frac{1}{4}\right)$ & $\left(2\frac{1}{4}, 2\frac{1}{4}\right)$\\
\hline
\xmark & $\left(2\frac{1}{4}, 2\frac{1}{4}\right)$ & $\left(2\frac{1}{4}, 2\frac{1}{4}\right)$ & $\left(2\frac{1}{4}, 2\frac{1}{4}\right)$\\
\hline
\end{tabular}
\caption{NE payoffs in the class $B$ extension for the standard PD (\ref{pdgame}). The mark \xmark \, denotes lack of a NE for the corresponding pair of strategies.}\label{Table_B2}
\end{center}
\end{table}
\end{example}
\subsection{Extension of the $C$ class}
In the following subsection, the $C$ class extension is discussed to analyse the possible NE. The payoff matrix of the first player is
\begingroup 
\setlength\arraycolsep{3pt}
\renewcommand\arraystretch{1.5}
\medmuskip = 0.6mu
\begin{align}
&C^1 =\begin{pmatrix} \label{C_matrix_1}
r & 0 &\frac{t}{2}  (p+r)+\frac{1-t}{2} & \frac{1}{2} (1-t) (p+r)+\frac{t}{2} \\
1 & p &\frac{1}{2} (1-t) (p+r)+\frac{t}{2}& \frac{t}{2}  (p+r)+\frac{1-t}{2}, \\
\frac{t}{2}  (p+r)+\frac{1-t}{2} & \frac{1}{2} (1-t) (p+r)+\frac{t}{2} & p t^2+r (1-t)^2+t (1-t) & (1-t) t (p+r)+t^2 \\
\frac{1}{2} (1-t) (p+r)+\frac{t}{2}& \frac{t}{2}  (p+r)+\frac{1-t}{2} & t (1-t) (p+r)+(1-t)^2& p (1-t)^2+r t^2+t (1-t)
\end{pmatrix}.
\end{align}
\endgroup
The existence of NE in four-strategy quantum game extensions of PD is dependent on $p$ and $r$, i.e. PD 
payoffs and on the EWL scheme parameter $t$ according to \eqref{klasaC}. We found that neither of the pair of strategies in the first row or the first column of the $C$ 
class could result in a NE. This can be proven by the following proposition. 

\begin{proposition} \label{C_prop_1}
  Neither $(1,j)$ nor $(i,1)$, $i,j \in \{1,2,3,4\}$, is a Nash equilibrium.
\end{proposition}
\begin{proof}
Note that none of the pair of strategies represented in the  first row or the first column of the $C$ class can be a NE because 
 \begin{align}
  & \frac{1}{2}(1-t)(p+r)+\frac{t}{2} < 1, \label{C1_first_r_c}\\
  & \frac{t}{2}(p+r)+\frac{1-t}{2}t <1\label{C2_first_r_c}
 \end{align}
for any $t\in(0,1)$. Indeed, after some transformations, inequality 
\eqref{C1_first_r_c}  takes on the form of $(p+r-1)t-(p+r-2)>0$. Since 
\begin{align}
& (p+r-1)t-(p+r-2)>(p+r-2)t-(p+r-2)=(p+r-2)(t-1)>0,
\end{align}
inequality 
\eqref{C1_first_r_c} remains true for $t\in(0,1)$. Simultaneously, 
inequality \eqref{C2_first_r_c} is equivalent to  $(p+r-1)t-1<0$. 
Since $(p+r-1)t-1<t-1$, then the solution of \eqref{C2_first_r_c} is given by $t \in (0,1)$.
\end{proof} 
The following propositions show the conditions that must be satisfied by the parameters $t$, $p$, and $r$ for the existence of NE in the remaining strategy profiles. 
\begin{proposition} \label{C_prop_2}
The strategy profile (2,2) is a Nash equilibrium if either $\frac{r-p}{p+r-1}\leq t\leq \frac{2p-1}{p+r-1}$ and $ p > \frac{1}{2}$ or $t=\frac{1}{2}$ and $p=\frac{1+r}{3}$ are the case.
\end{proposition}
\begin{proof}
Consider the following inequality
\begin{equation}
    p \geq \frac{1}{2}(1-t)(p+r)+\frac{t}{2}, 
\end{equation}
or, equivalently, $(p+r-1)t+p-r\geq 0$. Let $p>\frac{1}{2}$, then  $\frac{r-p}{p+r-1}\leq t<1$. Inequality  $p \geq \frac{t}{2}(p+r)+\frac{1-t}{2}$ or its equivalent form $(p+r-1)t+1-2p\leq0$ is fulfilled if $0<t\leq \frac{2p-1}{p+r-1}$ and $p>\frac{1}{2}$. 
Consequently, a pair of strategies given by the element $(2,2)$ of $C$ class is a NE when $\frac{r-p}{p+r-1}\leq t\leq \frac{2p-1}{p+r-1}$ with the assumption of $p>\frac{1+r}{3}$. Furthermore, if $p=\frac{1+r}{3}$, then  the pair of strategies given by the element $(2,2)$ results in a NE if $t=\frac{1}{2}$.
\end{proof}
\begin{proposition} \label{C_prop_3}
The strategy profile $(3,3)$ is a  Nash equilibrium of the $C$ game provided that $t=\tfrac{1}{2}$.
\end{proposition}
\begin{proof}
In what follows, we will prove that the intersection of the following set of inequalities 
\begin{align}
&p t^2+r (1-t)^2+t (1-t) \geq  \frac{t}{2}  (p+r)+\frac{1-t}{2}, \label{C_3.3_1}\\
&p t^2+r (1-t)^2+t (1-t) \geq \frac{1}{2} (1-t) (p+r)+\frac{t}{2}, \label{C_3.3_2}\\
&p t^2+r (1-t)^2+t (1-t) \geq  t (1-t) (p+r)+(1-t)^2, \label{C_3.3_3} 
\end{align}
is given by $t=\tfrac{1}{2}$. Inequality \eqref{C_3.3_1} is equivalent to $(2t-1)(t(p+r-1)-2r+1) \geq 0$. 
Since 
\begin{align}
& (p+r-1)t-(2r-1)<(2r-1)t-(2r-1)<(2r-1)(t-1)<0,   
\end{align}
and if  $t\in(0,1)$, then $2t-1\leq 0$, 
and hence, $t\in \left(0,\frac{1}{2}\right]$. Inequality \eqref{C_3.3_3} is equivalent to 
$(2t-1)(t(p+r-1)-(r-1))\geq 0 $. 
Since
\begin{align}\label{r-1a}
& (p+r-1)t-(r-1)>(r-1)t-(r-1)=(r-1)(t-1)>0,
\end{align}
this inequality  holds for every  $t \in \left[\frac{1}{2},1\right)$.

It is easy to check that the remaining inequality \eqref{C_3.3_2} is 
fulfilled for $t=\frac{1}{2}$.
Hence, we can conclude that a pair of strategies $(3,3)$ 
results in a NE, if $t=\frac{1}{2}$.
\end{proof}
\begin{proposition}
The strategy profile $(4,4)$ is a  Nash equilibrium of the $C$ game provided that $t=\tfrac{1}{2}$.
\end{proposition}
\begin{proof}
In what follows, we will prove that the intersection of the following set of inequalities:
 \begin{align}
&p (1-t)^2+r t^2+t (1-t) \geq \frac{1}{2} (1-t) (p+r)+\frac{t}{2} \label{C_4.4_1}\\
&p (1-t)^2+r t^2+t (1-t) \geq  \frac{t}{2}  (p+r)+\frac{1-t}{2} \label{C_4.4_2}\\
&p (1-t)^2+r t^2+t (1-t) \geq  (1-t) t (p+r)+t^2 \label{C_4.4_3} 
\end{align}
is given by $t=\frac{1}{2}$.

Inequality \eqref{C_4.4_1} can be rewritten in the form of $(2t-1)(t(p+r-1)+r-p)\geq 0$. One can note that
\begin{align}
&t(p+r-1) +r -p = tp + tr -t +r -p = p(t-1) + 2tr-tr-t+r \\
&=p(t-1) - r(t-1) + t(2r-1) = (p-r)(t-1) + t(2r-1) >0. \nonumber
\end{align}
Hence, one can conclude that
inequality \eqref{C_4.4_1} is fulfilled for  $t\in \left[\frac{1}{2},1\right)$.

Inequality \eqref{C_4.4_3} is equivalent to $(2t-1)(t(p+r-1)-p)\geq 0$. Note that
\begin{equation}\label{p+r-1}
 p+r-1<p, \text{ for every } t\in(0,1).
\end{equation}
Consequently, we can conclude that the solution of 
\eqref{C_4.4_3} is given by the following intersection 
$\left(0,\frac{1}{2}\right]$.

One can easily check that the remaining 
inequality \eqref{C_4.4_2} is
true
for $t=\frac{1}{2}$. Hence, we conclude that the pair of 
strategies given by the element  $(4,4)$ is a NE if $t=\frac{1}{2}$. 
\end{proof}
To prove that pairs of strategies $(3,2)$ and $(2,3)$ are NE of the $C$ class game one can provide the following proposition. 

\begin{proposition}
 Let  $0<p<1-r$. Strategy profiles $(3,2)$ and 
 $(2,3)$ are  Nash equilibria of the $C$ game if 
 $t\geq \tfrac{1}{2}$. Moreover, $(3,2)$ and $(2,3)$ are  Nash equilibria of the $C$ game if $t=\frac{1}{2}$ and $1-r<p<\frac{1+r}{3}$.
\end{proposition}
\begin{proof}
Consider the following set of inequalities:
\begin{align}
&  \frac{1}{2} (1-t) (p+r)+\frac{t}{2} \geq    \frac{t}{2} (p+r)+\frac{1-t}{2} \label{C_2.3_1}\\
&  \frac{1}{2} (1-t) (p+r)+\frac{t}{2} \geq p t^2+r (1-t)^2+t (1-t) \label{C_2.3_2}\\
&  \frac{1}{2} (1-t) (p+r)+\frac{t}{2} \geq t (1-t) (p+r)+(1-t)^2 \label{C_2.3_3}\\
&  \frac{1}{2} (1-t) (p+r)+\frac{t}{2} \geq p \label{C_2.3_5}
\end{align}

Inequality \eqref{C_2.3_3} or equivalently $(2t-1)(t(p+r-1)-(p+r-2)) \geq 0$, is satisfied for $t\in \left[\frac{1}{2},1\right)$.

Let now $\tfrac{1}{2}\leq t<1$. If $r<1-p$, then
inequality \eqref{C_2.3_1}, equivalent to $(2t-1)
((p+r-1)t+p-r)\leq 0$,  is satisfied.

Inequality \eqref{C_2.3_2} is equivalent to $(2t-1)
((p+r-1)t+p-r)\leq 0$,which is satisfied for 
$\tfrac{1}{2}\leq t<1$ and $r<1-p$. Moreover, inequality 
\eqref{C_2.3_5}, of the equivalent form of $(p+r-1)t+p-r\leq 0$, is also satisfied for  $\tfrac{1}{2}\leq  t<1$ and $r<1-p$.

Consequently, we infer that  if $0<p<1-r$ and $\frac{1}{2}\leq 0<1$, then $(3,2)$ and $(2,3)$ are  NE of the $C$ game.

It is easy to prove that if $t=\frac{1}{2}$ and $1-
r<p<\frac{1+r}{3}$, then $(3,2)$ and $(2,3)$ are  
NE of the $C$ game.
\end{proof}

\begin{proposition}
  Let  $0<p<1-r$. Pairs of strategies $(4,2)$ and $(2,4)$ are  Nash equilibria of the $C$ game if $t\leq\tfrac{1}{2}$. In particular, $(4,2)$ and $(2,4)$ are Nash equilibria of the $C$ game if $t=\frac{1}{2}$ and $1-r<p<\frac{1+r}{3}$.
\end{proposition}

\begin{proof}
Consider the following set of inequalities:
 \begin{align}
&\frac{t}{2}(p+r)+\frac{1-t}{2}\geq \frac{1}{2}(1-t)(p+r)+\frac{t}{2},  \label{C_4_2.1}\\
& \frac{t}{2}(p+r)+\frac{1-t}{2}\geq  (1-t)t(p+r)+t^2,  \label{C_4_2.2}\\
& \frac{t}{2}(p+r)+\frac{1-t}{2}\geq p(1-t)^2+rt^2+t(1-t)  \label{C_4_2.3}\\
& \frac{t}{2}(p+r)+\frac{1-t}{2}\geq p, \ \label{C_4_2.5}
\end{align}
One notes that inequality \eqref{C_4_2.2} is 
equivalent to $(2t-1)((p+r-1)t-1 \geq 0$ and 
is satisfied for $t\in \left(0,\frac{1}{2}\right]$.

Let $0<t\leq\tfrac{1}{2}$ and $r<1-p$. Then, 
inequality \eqref{C_4_2.1}, equivalent to $(2t-1)
(p+r-1)\geq 0$, is satisfied. Since
$t(p+r-1)+1-2p > t(2p-1)-(2p-1)=(2p-1)(t-1)$
and $2p-1\leq1-2r<0$, therefore inequality  
\eqref{C_4_2.5} is true for $t\in \left(0,\frac{1}
{2}\right]$. From this, we immediately get the 
solution of equation $(1-2t)(t(p+r-1)-(2p-1))\geq 
0$, equivalent to \eqref{C_4_2.3}, namely 
$\left(0,\frac{1}{2}\right]$.

It is easy to see that if  $t=\frac{1}{2}$ and $1-
r<p<\frac{1+r}{3}$, then $(4,2)$ and $(2,4)$ are  
NE of the $C$ game.
\end{proof}

The following reasoning can be provided to prove 
that pairs of strategies $(3,4)$ and $(4,3)$ are  
NE of the $C$ game under certain 
conditions.
\begin{proposition}
Strategy profiles $(3,4)$, $(4,3)$ are  Nash equilibria of the $C$ game provided that $t=\tfrac{1}{2}$.
\end{proposition}
\begin{proof}
Below it will be proved that the following set of 
inequalities is satisfied for $t=\frac{1}{2}$.

\begin{align}
&(1-t) t (p+r)+t^2 \geq \frac{1}{2} (1-t) (p+r)+\frac{t}{2}\label{C_3.4_1a} \\
&(1-t) t (p+r)+t^2 \geq  \frac{t}{2}  (p+r)+\frac{1-t}{2}\\
&(1-t) t (p+r)+t^2 \geq  p (1-t)^2+r t^2+t (1-t) \\
&t (1-t) (p+r)+(1-t)^2 \geq  \frac{t}{2}  (p+r)+\frac{1-t}{2} \\
&t (1-t) (p+r)+(1-t)^2 \geq \frac{1}{2} (1-t) (p+r)+\frac{t}{2} \label{C_4.3_2a}\\
&t (1-t) (p+r)+(1-t)^2 \geq  p t^2+r (1-t)^2+t (1-t)
\end{align}

Inequality \eqref{C_3.4_1a} is equivalent to $(2t-1)
((p+r)-t(p+r-1))\geq 0$. One notes that $t(p+r-1)-
(p+r) < 0$, therefore $2t-1\geq 0$, and hence $t\in 
\left[\frac{1}{2},1\right)$.

Inequality \eqref{C_4.3_2a} can be rewritten in the 
following form $(2t-1)((p+r-1)t-(p+r-2))\leq 0$. 
Since $(p+r-1)t-(p+r-2)>0$, then inequality 
\eqref{C_4.3_2a} is satisfied if $t\in(0,\frac{1}
{2}]$.

It can be easily proved that the remaining 
inequalities are satisfied for $t=\frac{1}{2}$. Hence, one can infer that pairs of strategies
$(3,4)$ and $(4,3)$ are NE for $t=\frac{1}{2}$. 
\end{proof}

The existence of NE in the $C$ class 
extension can be summarized in the following Table \ref{C_NE_param} where 
particular cells refer to the corresponding pairs of
strategies given in the $C$ class.

\begin{table}[H] 
\renewcommand{\arraystretch}{1.5}
    \centering
\begin{tabular}{ | >{\centering}p{0.9cm} | >{\centering}p{4.5cm} | >{\centering}p{4.5cm} | >{\centering\arraybackslash}  p{4.5cm} | }

    \hline
\multirow{2}{*} \xmark & \multirow{2}{*} \xmark & \multirow{2}{*} \xmark & \multirow{2}{*} \xmark\\  
 & & & \\
\hline
\multirow{2}{*}  \xmark    & $\left(p>\frac{1}{2}\right.$ $\wedge$ $\left.\frac{r-p}{p+r-1} \leq t \leq \frac{2p-1}{p+r-1} \right)$~$\vee$     &   {{$\left(0<p\leq 1-r\right.$ $\wedge$ $\left.t\geq 1/2\right)$~$\vee$ } }     &  {{  $\left(0<p\leq 1-r\right.$ $\wedge$ $\left.t\leq 1/2\right)$~$\vee$  }}        \\

                    &   {$\left(p = \frac{r+1}{3}\right.$ $\wedge$ $\left.t=\frac{1}{2}\right)$}   &  $\left(1-r<p\leq \frac{1+r}{3}\right.$ $\wedge$ $\left.t=\frac{1}{2}\right)$     &  $\left(1-r<p\leq \frac{1+r}{3}\right.$ $\wedge$ $\left.t=\frac{1}{2}\right)$   \\
    \hline
    \multirow{2}{*} \xmark    &   $\left(0<p \leq 1-r\right.$ $\wedge$ $\left.t \geq 1/2\right)$~$\vee$       &  \multirow{2}{*}{ $0<p<r<1 \, \wedge \, 2r > 1\,\wedge \,  t=\frac{1}{2}$  }  &   \multirow{2}{*}{ $0<p<r<1 \, \wedge \, 2r > 1\,\wedge \,  t=\frac{1}{2}$  }         \\

            &    $\left(1-r< p \leq \frac{r+1}{3}\right.$ $\wedge$ $\left.t=\frac{1}{2}\right)$    &    &   \\
    \hline
     \multirow{2}{*} \xmark     &   $\left(0<p\leq 1-r\right.$ $\wedge$ $\left.t\leq 1/2\right)$~$\vee$     &   \multirow{2}{*}{ $0<p<r<1 \, \wedge \, 2r > 1\,\wedge \,  t=\frac{1}{2}$  }    &   \multirow{2}{*}{ $0<p<r<1 \, \wedge \, 2r > 1\,\wedge \,  t=\frac{1}{2}$  }       \\

              &  $\left(1-r< p \leq \frac{r+1}{3}\right.$  $\wedge$ $\left.t=\frac{1}{2}\right)$    &        &    \\ \hline
\end{tabular}
\caption{The $C$ class strategies parameters resulting in NE. Parameters $p,\,r$ are directly related to PD payoffs (\ref{afinematrix}), while $t$ refers to EWL scheme parameter $\theta_1$ (\ref{klasaC}). The mark \xmark \, denotes lack of a NE for the corresponding pair of strategies.}\label{C_NE_param}
\end{table}
\begin{example}
 Consider the PD given by Eq. \eqref{pdgame}.
The $C$ class extension takes on the following 
form:

\begingroup 
\setlength\arraycolsep{1.7pt}
\renewcommand\arraystretch{1.5}
\begin{align}
&C =\begin{pmatrix} 
(3,3) & (0,5) & \left(\frac{5-t}{2},\frac{5-t}{2}\right) & \left(\frac{4+t}{2},\frac{4+t}{2}\right) \\
(5,0) & (1,1) &\left(\frac{4+t}{2},\frac{4+t}{2}\right)& \left(\frac{5-t}{2},\frac{5-t}{2}\right) \\
\left(\frac{5-t}{2},\frac{5-t}{2}\right) & \left(\frac{4+t}{2},\frac{4+t}{2}\right)  & \left(3-t-t^2,3-t-t^2\right) & \left(t(t+4),5-6t+t^2\right) \\
\left(\frac{4+t}{2},\frac{4+t}{2}\right) & \left(\frac{5-t}{2},\frac{5-t}{2}\right)  & \left(5-6t+t^2,t(t+4)\right) & \left(-t^2+3 t+1,-t^2+3 t+1\right)
\end{pmatrix}
\end{align}
\endgroup

We note that in the following form of PD the 
strategy profile (2,2) does 
not satisfy the condition $p > \frac{1}{2}$ since 
$p=\frac{1}{5}$. Hence, (2,2) is not a NE. 

Therefore, there are eight pure NE denoted graphically in 
 Table \ref{C_PD_NE_payoffs}. The maximum payoff for both
players is equal to $2\frac{1}{2}$ and is achieved for 
two pairs of strategies, namely (2,3) and 
(3,2). 

Please note that the maximum 
payoff of these two pairs of strategies 
was obtained for the upper limit of $t$ 
which in this case is equal to 1 (see 
Table \ref{C_NE_param}).

\begin{table}[ht]
\renewcommand{\arraystretch}{1.7}
\begin{center}
\begin{tabular}{ | >{\centering}p{1.5cm} | >{\centering}p{1.5cm} | >{\centering}p{1.5cm} | >{\centering\arraybackslash}  p{1.5cm} | }
\hline
\xmark & \xmark & \xmark & \xmark\\
\hline
\xmark & \xmark & $\left(2\frac{1}{2}, 2\frac{1}{2}\right)$  & $\left(2\frac{1}{4}, 2\frac{1}{4}\right)$ \\
\hline
\xmark & $\left(2\frac{1}{2}, 2\frac{1}{2}\right)$  & $\left(2\frac{1}{4}, 2\frac{1}{4}\right)$ & $\left(2\frac{1}{4}, 2\frac{1}{4}\right)$\\
\hline
\xmark & $\left(2\frac{1}{4}, 2\frac{1}{4}\right)$ & $\left(2\frac{1}{4}, 2\frac{1}{4}\right)$ & $\left(2\frac{1}{4}, 2\frac{1}{4}\right)$\\
\hline
\end{tabular}
\caption{NE payoffs in the class $C$ extension for the standard PD (\ref{pdgame}). The mark \xmark \, denotes lack of a NE for the corresponding pair of strategies.}\label{C_PD_NE_payoffs}
\end{center}
\end{table}

Additionally, one can provide a following 
example of PD payoffs that would result in 
the strategy profile (2,2) being a NE, 
namely $r=4/5$ and $p=3/5$. Such pair of 
parameters satisfies the condition $p=(r+1)/3$ discussed in Proposition  \ref{C_prop_2}
and with $t=1/2$ 
results in a NE.

One notes that for $t=1/2$ the $C$ class extension 
given by Eq. \eqref{C_matrix_1} is equivalent to 
the $B$ class one  written in the form of Eq. 
\eqref{B_ext}.
\end{example}

\subsection{Extension of the $D$ class}

Let $D_1=\left(D_1^1,\left(D_1^1\right)^T\right)$, where
\begingroup 
\setlength\arraycolsep{3pt}
\renewcommand\arraystretch{1}
\medmuskip = 0.6mu
\begin{equation}
 D_1^1=\begin{pmatrix}
      r & 0 &rt & r-rt\\
      1 & p& (1-p)t+p & (p-1)t+1\\
     (r-1)t+1 & p-pt& (r-1)t^2+t+p(1-t)^2&(p+r)(t-t^2)+(1-t)^2\\
     (1-r)t+r & pt & (1-p-r)t^2+(p+r)t&(p-1)t^2+t+r(1-t)^2
     \end{pmatrix}.
\end{equation}
\endgroup
\begin{proposition}
    Let $t\in(0,1)$. Strategy profile $(2,2)$ is a unique Nash equilibrium of $D_1$ game for all $p$, $r$.
\end{proposition}
 \begin{proof}
It is easily seen that
\begin{equation}\label{pp}
 \begin{gathered}
 p> p-pt,\\
 p> pt
\end{gathered}
\end{equation}
for every
 $t\in(0,1)$. Therefore,  a profile strategy   $(2,2)$  is a NE for every $t\in(0,1)$.

Since the inequalities $(r-1)t+1<1$, $(1-r)t<1$ are satisfied for $t\in(0,1)$, we conclude that neither $(1,j)$ nor $(i,1)$, $i,j\in\{1,2,3,4\}$, is a NE.

We aim to prove the following inequalities:
\begin{gather}
rt<(1-p)t+p,\label{k3_1}\\
 (r - 1) t^2 + t + p (1 - t)^2 < (1 - p) t + p, \label{k3_2}\\
 (1 - p - r) t^2 + (p + r) t < (1 - p) t + p \label{k3_3}
\end{gather}
for every $t\in(0,1)$.
The inequality \eqref{k3_1}  is equivalent to $(r+p-1)t<p$.  From  \eqref{p+r-1} we have
\begin{equation}\label{k3_4}
 (r+p-1)t-p<pt-p=p(t-1)<0.
\end{equation}
Hence, \eqref{k3_1} holds for every $t\in(0,1)$. The inequality \eqref{k3_2} is equivalent to
$t\left((r+p-1)t-p\right)<0$.  From \eqref{k3_4} we conclude that \eqref{k3_2} holds for every  $t\in(0,1)$. The inequality \eqref{k3_3} is equivalent to
\begin{equation}
 -(t-1)\left((p+r-1)t-p\right)<0.
\end{equation}
From \eqref{k3_4} it follows that $0<t<1$. From \eqref{pp},  \eqref{k3_1}-\eqref{k3_3} we conclude that  neither $(3,j)$ nor $(i,3)$, $i,j\in\{1,2,3,4\}$, is a NE.

Now we show that
\begin{equation}\label{d2_44a}
 (p-1)t^2+t+r(1-t)^2<(p-1)t+1
\end{equation}
or, equivalently,
\begin{equation}
 (t-1)\left((p+r-1)t-(r-1)\right)<0
\end{equation}
holds for every $t\in(0,1)$. From \eqref{r-1a}
the inequality \eqref{d2_44a} holds for every $t\in(0,1)$, and  a strategy profile $(4,4)$ is not NE.
\end{proof}
\begin{example}
 Consider the PD \eqref{pdgame}.
Then
\begin{equation}
 D_1=\begin{pmatrix}
     (3,3) & (0,5) & (3t,5-2t) & (3-3t,2t+3)\\
         (5,0) & (1,1)& (4t+1,1-t) & (5-4t, t)\\
         (5-2t,3t)& (1-t,4t+1)& (1-t^2+3t, 1-t^2+3t)&(t^2-6t+5,t^2+4t)\\
         (2t+3,3-3t) & (t, 5-4t) &(t^2+4t, t^2-6t+5) & (3-t^2-t,3-t^2-t)
    \end{pmatrix}.
\end{equation}
We see that there is exactly one pure NE at a pair of strategies $(2,2)$. 
\end{example}
Let $D_2=\left(D_2^1,\left(D_2^1\right)^T\right)$, where
\begingroup 
\setlength\arraycolsep{3pt}
\renewcommand\arraystretch{1}
\medmuskip = 0.6mu
\begin{equation}
 D_2^1=\begin{pmatrix}
      r & 0 &(p-1)t+1 & (1-p)t+p\\
      1 & p& r-rt & rt\\
     pt & (1-r)t+r& (r-1)t^2+t+p(1-t)^2&(p+r)(t-t^2)+(1-t)^2\\
     p-pt& (r-1)t+1 & (1-p-r)t^2+(p+r)t&(p-1)t^2+t+r(1-t)^2
     \end{pmatrix}.
\end{equation}
\endgroup
\begin{proposition}
    Game $D_2$ has no Nash equilibria in pure strategies.
\end{proposition}
\begin{proof}
Since $pt\leq1$ and $p-pt\leq1$ hold for every $t\in(0,1)$, it follows that neither $(1,j)$ nor $(i,1)$, $i,j\in\{1,2,3,4\}$, is a NE.

The inequality
\begin{equation}
 rt<(1-p)t+p
\end{equation}
is equivalent to $-\left((p+r-1)t-p\right)>0$. We conclude from  \eqref{k3_4} that  the interval $(0,1)$ is solution of this inequality, and finally neither   $(4,2)$ nor $(2,4)$ is a NE.

It is easily to seen that
\begin{equation}
 (r-1)t-(p-1)>(p-1)t-(p-1)=(p-1)(t-1)>0
\end{equation}
for every $t\in(0,1)$. Hence  a strategy profile $(2,2)$
is not a NE.

Consider
\begin{equation}
 r-rt<(p-1)t+1
\end{equation}
or equivalently, $(p+r-1)t+(1-r)>0$. From \eqref{r-1a}   it follows that  $t\in(0,1)$, hence neither $(3,2)$ nor $(2,3)$ is a NE. We conclude from \eqref{d2_44a} and  $(p-1)t+1<(r-1)t+1$ for every $t\in(0,1)$ that a strategy profile $(4,4)$ is not a NE. A strategy profile $(3,3)$ is not a NE, which follows from inequalities
\begin{equation}\label{d2_33a}
 (r-1)t^2+t+p(1-t)^2\geq(1-p-r)t^2+(p+r)t \text{ for every } t\in\left(0,\tfrac{1}{2}\right],
\end{equation}
\begin{equation}\label{d2_33b}
 (r-1)t^2+t+p(1-t)^2\leq(p-1)t+1 \text{ for every } t\in\left(0,\tfrac{1}{2}\right].
\end{equation}
It easily seen that the solution of \eqref{d2_33a}, or equivalently   $\left(t-\frac{1}{2}\right)\left((r+p-1)t-p\right)<0$ is  $\left(0,\tfrac{1}{2}\right]$. Note that the solution of $(1-p)t+p\leq(p-1)t+1$  is $\left(0,\tfrac{1}{2}\right]$. We conclude from this and  from \eqref{k3_3} that \eqref{d2_33b} holds for every $t\in\left(0,\tfrac{1}{2}\right]$.

Since the following inequalities
\begin{equation}\label{d2_43a}
 (1-p-r)t^2+(p+r)t\geq(r-1)t^2+t+p(1-t)^2,
\end{equation}
\begin{equation}\label{d2_43b}
 (p+r)(t-t^2)+(1-t)^2\leq(p-1)t^2+t+r(1-t)^2
\end{equation}
hold for every $t\in\left[\tfrac{1}{2},1\right)$, it follows that neither  $(4,3)$ nor $(3,4)$ is a NE.  Consider \eqref{d2_43a}, or equivalently, $(-1+2t)\left((p+r-1)t-p\right)\leq 0$. From \eqref{k3_4} it follows that $t\in\left[\tfrac{1}{2},1\right)$.
Moreover, from \eqref{r-1a}  we conclude that  \eqref{d2_43b}, or equivalently  $(2t-1)\left((p+r-1)t+1-r\right)\geq 0$ holds for every $t\in\left[\tfrac{1}{2},1\right)$.
\end{proof}
\begin{example}
 Consider the PD  \eqref{pdgame}.
Then
\begin{equation}
 D_2=\begin{pmatrix}
     (3,3) & (0,5) & (5-4t,t) & (1+4t,1-t)\\
         (5,0) & (1,1)& ({3-3t},3+2t) & (3t, 5-2t)\\
         (t,5-4t)& (3+2t,3-3t)& (1-t^2+3t, 1-t^2+3t)&(t^2-6t+5,t^2+4t)\\
         (1-t,1+4t) & (5-2t,3t) &(t^2+4t, t^2-6t+5) & (3-t^2-t,3-t^2-t)
    \end{pmatrix}.
\end{equation}
It is easy to check that this game has no NE in pure strategies.
\end{example}

 \subsection{Extension of the $E$ class}
 
Let $E_1=\left(E_1^1,\left(E_1^1\right)^T\right)$, where
\begingroup 
\setlength\arraycolsep{3pt}
\renewcommand\arraystretch{1}
\medmuskip = 0.6mu
\begin{equation}
 E_1^1=\begin{pmatrix}
      r & 0 &(r-1)t+1 & p-pt\\
      1 & p& (1-r)t+r & pt\\
     rt & r-rt& (r-1)t^2+t+p(1-t)^2&(p+r)(t-t^2)+(1-t)^2\\
     (1-p)t+p & (p-1)t+1 & (1-p-r)t^2+(p+r)t&(p-1)t^2+t+r(1-t)^2
     \end{pmatrix}.
\end{equation}
\endgroup
\begin{proposition}
 Let $\frac{1}{2}\leq t <1$. Then a strategy profile $(4,4)$ is a unique Nash equilibrium of $E_1$ game for all  $p$, $r$.
\end{proposition}
\begin{proof}
Let us first prove that the following inequalities hold:
\begin{gather}
 (p-1)t^2+t+r(1-t)^2\geq p-pt \text{ for every } \tfrac{1}{2}\leq t<1,\label{e1_44a}\\ 
 (p-1)t^2+t+r(1-t)^2\geq pt\text{ for every } 0\leq t<1,\label{e1_44b}\\
 (p-1)t^2+t+r(1-t)^2 \geq (p+r)(t-t^2)+(1-t)^2 \text{ for every } \tfrac{1}{2}\leq t<1.\label{e1_44c}
\end{gather}
The inequality \eqref{e1_44b} is equivalent to $(t-1)\left((p+r-1)t-r\right)\geq 0$. Clearly
\begin{equation}
 (p+r-1)t-r<rt-r=r(t-1)<0,
\end{equation}
hence that \eqref{e1_44b}  holds for every  $t\in(0,1)$. 
From \eqref{r-1a}  it follows that \eqref{e1_44c}, or equivalently \begin{equation}
(1-2t)\left((p+r-1)t+r-1\right)\geq 0\end{equation} holds for every $t\in \left[\tfrac{1}{2},1\right)$. 
Since $p(1-t)\leq pt$ for every $t\in\left[\tfrac{1}{2},1\right)$,   the solution of \eqref{e1_44a} is $\left[\tfrac{1}{2},1\right)$. 
From \eqref{e1_44a}-\eqref{e1_44c} we conclude that a strategy profile  $(4,4)$ is a NE for every $t\in\left[\tfrac{1}{2},1\right)$.

It is clear that the following inequalities hold for every $t\in(0,1)$:
\begin{equation}
 \begin{gathered}
  (1-p)t+p< 1,\\
  rt< 1.
 \end{gathered}
\end{equation}
Hence neither $(1,j)$ nor $(i,1)$, $i,j\in\{1,2,3,4\}$, is a NE. 
From \eqref{e1_44b} we see that strategy profiles $(4,2)$ and $(2,4)$ are not NE. 
From \eqref{r-1a}  it follows that
\begin{equation}
 (p+r-1)t-r+1\geq 0
\end{equation}
holds for every $t\in(0,1)$. Hence strategy profiles $(3,2)$, $(2,3)$ are not a NE. 
From \eqref{d2_43a} and \eqref{d2_43b} we conclude that neither  $(3,4)$ nor $(4,3)$  is a NE.

Since
\begin{equation}\label{e1_33}
 (r-1)t^2+t+p(1-t)^2\leq (r-1)t+1
\end{equation}
for every $t\in(0,1)$, it follows that a strategy profile $(3,3)$ is not a NE. Indeed, \eqref{e1_33} is equivalent to
$(t-1)\left((p+r-1)t+1-p\right)\leq 0.$
 It follows easily that
 \begin{equation}
  (p+r-1)t+1-p>(p-1)t-(p-1)=(p-1)(t-1)>0,
 \end{equation}
hence $t\in(0,1)$.
\end{proof}
\begin{example}
 Consider the PD  \eqref{pdgame}.
Then
\begin{equation}
 E_1=\begin{pmatrix}
     (3,3) & (0,{5}) & (5-2t,3t)& (1-t,{4t+1})\\
         ({5},0) & ({1},{1})& (2t+3,3-3t) & (t,{5-4t})\\
         (3t,5-2t) & (3-3t,2t+3)& (1-t^2+3t, 1-t^2+3t)&(t^2-6t+5,t^2+4t)\\
         ({4t+1},1-t) & ({5-4t}, t) &(t^2+4t, t^2-6t+5) & (3-t^2-t,3-t^2-t)
    \end{pmatrix}.
\end{equation}
There is exactly one pure NE at a pair of strategies $(4,4)$. The maximum payoff for both players is reached at $t = 1/2$, is the same and is equal to $2\frac{1}{4}$. 
\end{example}
Let $E_2=\left(E_2^1,\left(E_2^1\right)^T\right)$, where
\begingroup 
\setlength\arraycolsep{3pt}
\renewcommand\arraystretch{1}
\medmuskip = 0.6mu
\begin{equation}
 E_2^1=\begin{pmatrix}
      r & 0 &pt & (1-r)t+r\\
      1 & p& p-pt & (r-1)t+1\\
     (p-1)t+1 & (1-p)t+p& (r-1)t^2+t+p(1-t)^2&(p+r)(t-t^2)+(1-t)^2\\
     r-rt & rt & (1-p-r)t^2+(p+r)t&(p-1)t^2+t+r(1-t)^2
     \end{pmatrix}.
\end{equation}
\endgroup
\begin{proposition}
Let $0<t \leq \frac{1}{2}$. Then a strategy profile $(3,3)$ is a unique Nash equilibrium of $E_2$ game.
\end{proposition}
\begin{proof}
Let us first prove that the following inequalities hold:
\begin{gather}
 (r-1)t^2+t+p(1-t)^2\geq p-pt \text{ for every } 0< t<1,\label{e2_33a}\\
 (r-1)t^2+t+p(1-t)^2\geq pt\text{ for every } 0<t\leq \tfrac{1}{2},\label{e2_33b}\\
 (r-1)t^2+t+p(1-t)^2 \geq (1-p-r)t^2+(p+r)t \text{ for every } 0<t\leq \tfrac{1}{2}.\label{e2_33c}
\end{gather}
The inequality \eqref{e2_33a} is equivalent to  $t\left((p+r-1)t+1-p\right)\geq 0.$ Since $p+r-1>r-1$  and $p-1<r-1$,
\begin{equation}
 (p+r-1)t-(p-1)>(p-1)t-(p-1)=(p-1)(t-1)>0.
\end{equation}
Hence $t\in(0,1)$. Consider \eqref{e2_33c}, or equivalently  $(2t-1)\left((p+r-1)t-p\right)\geq 0$.
From  \eqref{p+r-1} we have
\begin{equation}\label{e2_33d}
 (p+r-1)t-p<pt-p=p(t-1)<0.
\end{equation}
It follows that $t\in\left(0,\tfrac{1}{2}\right]$. Note that $pt\leq p(1-t)$ where $t\in\left(0,\tfrac{1}{2}\right]$. Hence
\begin{equation}
 (r-1)t^2+t+p(1-t)^2\geq p-pt\geq pt \text{ for every } \tfrac{1}{2}\leq t<1.
\end{equation}
From \eqref{e2_33a}-\eqref{e2_33c} we conclude that a strategy profile $(3,3)$ is a NE for every $t\leq \tfrac{1}{2}$.

It is easily seen that the following inequalities hold for every $t\in(0,1)$:
\begin{equation}
 \begin{gathered}
  (p-1)t+1< 1\\
  r-rt< 1.
 \end{gathered}
\end{equation}
Hence neither $(1,j)$ nor $(i,1)$, $i,j\in\{1,2,3,4\}$, is a NE.
Since $p\leq (1-p)t+p$ for every $t\in(0,1)$, a strategy profile $(2,2)$ is not a NE.
From \eqref{e2_33a} neither $(3,2)$ nor $(2,3)$ is a NE.
Clearly, from  \eqref{e2_33d} it follows that $(4,2)$ and $(2,4)$ are not NE \eqref{e2_33d}. Moreover, from \eqref{d2_43a}, \eqref{d2_43b} we conclude that  strategy profiles $(4,3)$, $(3,4)$ are not NE.
The following inequality holds for every $t\in(0,1)$:
\begin{equation}\label{e2_44}
 (p-1)t^2 +t +r(1-t)^2 \leq (r-1)t+1
\end{equation}
Indeed, from  \eqref{d2_44a} and $(p-1)t+1<(r-1)t+1$, it follows that \eqref{e2_44} holds. Hence a strategy profile $(4,4)$ is not a NE.
\end{proof}
\begin{example}
 Consider the PD  \eqref{pdgame}.
Then
\begin{equation}
 E_2=\begin{pmatrix}
     (3,3) & (0,{5}) & (t,{5-4t})& (2t+3,3-3t)\\
         ({5},0) & ({1},{1})& (1-t,{4t+1}) & (5t-2,3t)\\
         (5-4t,t) & (4t+1,1-t)& (1-t^2+3t, 1-t^2+3t)&(t^2-6t+5,t^2+4t)\\
         (3-3t,2t+3) & (3t,5-2t) &(t^2+4t, t^2-6t+5) & (3-t^2-t,3-t^2-t)
    \end{pmatrix}.
\end{equation}
There is exactly one pure NE at a pair of strategies $(3,3)$. 
Here again the maximum payoff is the same for both players, is reached at $t = 1/2$ and is equal to $2\frac{1}{4}$.
\end{example}
\section{Conclusions}
The nature of quantum games theory is highly complex as it combines several fields of 
science, namely
physics, computer science, mathematics, and 
economics. This results into a high  entry 
threshold for researchers who want to dive 
into this area. Yet, parallel to the 
developments in the field of emerging 
technologies, in particular quantum computing, 
advances in understanding of threats and 
profits of them are gathering significant 
importance. What is most commonly understood 
by enterprises are the strategic implications 
of quantum computing with respect to the 
impact on security. On top of that, there are 
also profits foreseen regarding the speed up 
of calculations. There are research 
institutions in the world that support 
enterprises with a roadmap that outlines 
milestones for information security teams 
when dealing with pending quantum threat. 

On the other hand, quantum key distribution promises
ultra-secure key distribution and via pilot projects 
in, for instance, critical systems  (e.g. electrical 
gird operations), are tested to provide basis for 
user data encryption.

With background of quantum advances recognised
rather as
posing a threat than bringing new chances we 
propose a quantum games theory approach where 
combining classical and quantum strategies 
opens a set of new possibilities for players 
to achieve their goals. The ultimate goal of applying 
quantum strategy is to improve the expected payoff of
individuals.

The objective of our research was to examine the quantum extensions of the general form of the Prisoner's Dilemma game, derived by incorporating two unitary strategies into the classical game \cite{frackiewicz_permissible_2024}. In this study, we identified all potential combinations of quantum strategies that result in Nash equilibria in pure strategies. These equilibria are observed in all possible extension classes with the exception of $D_1$ and $D_2$ classes. The conditions for the existence of the aforementioned equilibria are, in general, highly complex and include a number of relations between the payoffs of the classical game and certain parameters ($\theta_i$ and $\alpha_i$) of the unitary strategies. Furthermore, we investigated the value of equal payoffs of NE for extensions of the standard PD defined by equation (\ref{pdgame}). Our findings demonstrate that these payoffs obtain a maximal value of $5/2$, thereby approaching Pareto-optimal solutions more closely than the classical Nash equilibrium of the PD, which is equal to $1$. However, it proved impossible to obtain Pareto-optimal values, which in this case are equal to $3$. 

The findings can be employed as a basis for further investigation of NE, which can also be expressed in the form of mixed strategies. An intriguing avenue for research would be to ascertain whether such NE can be more closely aligned with Pareto-optimal solutions than the values obtained through the use of pure strategies.

\bibliographystyle{qipstyle}
\bibliography{references-3}

\end{document}